%% file: PBroadcastCovv2.tex
\documentclass[onecolumn,draftclsnofoot,12pt]{IEEEtran}
\usepackage{amssymb}
\usepackage{graphicx}
\usepackage{float}
\usepackage{graphics}
\usepackage{amsmath,amsthm}
\usepackage{epstopdf}
\usepackage{etoolbox}
\newtoggle{SingleColumn}
\toggletrue{SingleColumn}
\graphicspath{{./figs/}{./}{.//}}
\newcommand{\plotsize}{3.2in}
\IEEEoverridecommandlockouts
\title{Coverage Analysis of Broadcast Networks with Users Having Heterogeneous Content/Advertisement Preferences}
\author{Kanchan K. Chaurasia, Reena Sahu, Abhishek K. Gupta
\thanks{ The authors are with the Department of Electrical Engineering, Indian Institute of Technology Kanpur, Kanpur, India 208016. Email: grkabhi@iitk.ac.in. 
} }

\date{}

\newcommand{\addi}[1]{{\color{black}#1}}
\newcommand{\sym}{\mathrm{s}}
\input{defis1.tex}

\begin{document}

\maketitle

\begin{abstract}
This work is focused on the system-level performance of a broadcast network. Since all transmitters in a broadcast network transmit the identical signal, received signals from multiple transmitters can be combined to improve system performance. We develop a stochastic geometry based analytical framework to derive the coverage of a typical receiver. We show that there may exist an optimal connectivity radius that maximizes the rate coverage. Our analysis includes the fact that users may have their individual content/advertisement preferences. We assume that there are multiple classes of users with each user class prefers a particular type of content/advertisements and the users will pay the network only when then can see content aligned with their interest. The operator may choose to transmit multiple contents simultaneously to cater more users’ interests to increase its revenue. We present revenue models to study the impact of the number of contents on the operator revenue. We consider two scenarios for users' distribution- one where users' interest depends on their geographical location and the one where it doesn’t. With the help of numerical results and analysis, we show the impact of various parameters including content granularity, connectivity radius, and rate threshold and present important design insights.

\end{abstract}


\section{Introduction}
Broadcasting networks  provide society with various services including TV communication, delivery of critical information and alerts,  general entertainment, and educational services and thus have been a key wireless technology. 
With the recent advancement in wireless technologies and handheld electronic devices including smartphone and tablets, the use of broadcasting services has been extended to include many modern applications including delivery of traffic information to vehicles in vehicle-to-infrastructure networks, advertisement industry and mobile TV services. 
Many broadcasting standards have been recently proposed 
including the digital video broadcast-terrestrial standard (DVB-T2), the advanced television systems committee standard (ATSC 3.0), and  the DVB-next generation handheld standard (DVB-NGH) to assist delivering of TV broadcasting services to mobile devices \cite{ahu2}\cite{ahu3}\cite{ahu4}. 
 The advent of digital broadcasting has lead to a significant increase in the demand for multimedia services for handheld devices including mobile TV, live video streaming, and on-demand video in the last decade \cite{ahu2}. From these applications' perspective, broadcasting based multimedia services can provide better data rate and performance compared to the uni-cast cellular network based mobile TV.
Note that in a cellular network  where the desired data is transmitted to each user via orthogonal resources,  users may suffer from spectral congestion in regions with high density due to limit bandwidth resulting in the performance degradation. But, in a
broadcasting network providing multimedia services,  all transmitters transmit identical data to all  users and hence, do not require orthogonal resource. 
In these networks, each transmitter can use the complete spectrum to serve their users, and hence, these are also  called  single frequency networks (SFN).  Due to this, users may experience a better quality of service.  


\subsection{Related Work} 
Given increasing demand of broadcasting services, it is very interesting to analyze the broadcast
networks in terms of signal-to-noise-ratio (SINR) and achievable data rate to understand their limitations and potential to meet these demands. 
There have been some recent works in the system-level analysis of broadcast networks. 
 In \cite{ahu19}, the authors evaluated the  blocking probability  for users accessing a network delivering mobile TV services over a hybrid broadcast unicast   communication.
%
%
In \cite{ahu14}, the authors studied a cellular network with uni-cast and multicast-broadcast deployments. 
%
%
%
%
%
%
%
%
However, these works didn't include the effect of transmitters' locations in the evaluations which is required for the system-level analysis of broadcast networks. 
Stochastic geometry framework can be utilized to analyze  wireless networks from system level perspective    \cite{JGA2011, Primer2016, ahu7}.
Stochastic geometry based models have been validated for various types of networks including cellular networks,and  ad hoc networks \cite{Guo2015,JGA2011, AndrewsmmWaveTut2016,ahu8}. 
 In \cite{ahu9}, the authors describe the analytical approach to calculate the coverage probability of a hybrid broadcast and uni-cast network, however, the authors have only considered a single broadcast transmitter along with many uni-cast transmitters. 
%
%
%
%
%
%
%
 In our past work \cite{SaChGup2020}, we have considered a broadcast network with multiple broadcasting transmitters to compute the coverage performance of users. However, the work assumed a static connectivity region around the user where transmitters need to be located to be able to serve the users. As shown in this paper,  this connectivity region is of variable size depending on the location of the first closest transmitter. To the best of our knowledge, there exists no other past work which analyzes the SINR and rate performance of a broadcast network with multiple broadcasting
transmitters which is one of the main focuses of this paper.

Another important metric to evaluate broadcast networks is the revenue earned by the network operator.  
In a broadcast network, the revenue is generated either from subscribers as network access fees for viewing content of their choice or from advertisers to show their advertisements to interested subscribers. 
With digital broadcast, subscriptions and user-targeted advertisements added a new dimension in the revenue. 
Due to advancements in technologies over the past few decades,  the advertising has become more user targeted and location-adaptive which can be  planned according to the user demographics and their preferences to improve network revenue. 
In \cite{ahu23}, the authors studied location-based mobile marketing and advertising to show the positive interest of mobile consumers  in receiving relevant promotions. It is intuitive that  a targeted and localized content will have a better engagement factor. It is interesting to analyze the network revenue earned from the users with their preference dependent on their choices and geographical location. As far we know, there does not exist any past work that analyses the network revenue of a broadcast network with subscribers having preferences for content and  advertisement which is another focus of this paper.

\subsection{Contributions} 
In this paper, we derive an analytical framework to evaluate the performance of a broadcast network with multiple broadcasting transmitters with users having content preference. We also present a revenue model to quantify the network revenue earned by the network operator.  In particular, the contributions of this paper are as follows:
\begin{enumerate}
\item We consider a broadcast network with multiple transmitters. Since all transmitters in a broadcast network are transmitting the same signal, received signals from multiple transmitters from a certain connectivity region around the user can be combined to improve the coverage at this user.  Using tools from stochastic geometry, we derive the expression for SINR coverage and rate coverage of a typical receiver located at the origin. Due to the contribution in the desired signal power from multiple transmitters, the analysis is significantly different and difficult than their cellular counterpart. Our main contribution lies in developing the framework and deriving techniques to evaluate the analytical expressions of SINR and rate coverage. We show that this connectivity region depends on network bandwidth.
\item  We present some numerical results to validate our analysis and present design insights. We show the impact of connectivity region size, path-loss exponents, and the network density on the SINR and rate coverage. We also find that there exists an optimal size of connectivity region that maximizes the rate coverage. 
\item In this paper, we also include the fact that users may have their individual content or advertisement preferences. 
We assume that there are multiple classes of users with each class of users prefers a particular type of content/advertisements and the users will pay the network only when then can see a particular content of their interest. We assume that one unit of revenue comes to the network from a particular class of users if every user of this class can see the content as per the preference of this class.    We  study the revenue thus obtained by the network from users. The broadcast operator may choose to transmit multiple contents simultaneously to cater more users' interest to increase its revenue. However, given the limited resources, the network can cater only to few classes and this capability depends on how these user classes are distributed spatially. There are two scenarios considered for users' distribution. In one scenario, users' interest depend on their geographical position in the network and in the second scenario it does not. We calculate the analytical expression for SINR coverage and rate coverage at a typical user and evaluate the total revenue. We present many important design insights via numerical results.
\end{enumerate}

\textbf{Notation:} Let $\Ball(\x,r)$ denote the ball of radius $r$ with center at $\x$. $\dist{\x}$ denotes the norm of the vector $\x$ and $\dist{\x_i}=r_i$ denotes the random distance of BBS located at ${\x_i}$.  Let $\origin$ denote the origin. $\betafunc{z}{x}{y}$ is the incomplete Beta function which is defined as
 $$\betafunc{z}{x}{y}=\int_{0}^{z}u^{x-1}(1-u)^{y-1}\dd u.$$ Let $c$ denote the speed of EM waves in the  media.  $\complementT{\set{A}}$ denotes the complement of set $\set{A}$.

\section{System Model}

In this paper, we consider a broadcast network with multiple broadcasting base stations (BBSs), deployed in the 2D region $\set{R}=\mathbb{R}^2$. The considered system model is as follows:

 \begin{figure}[ht!]
    \centering
    \includegraphics[width=0.45\textwidth,trim=185 37 158 16,clip]{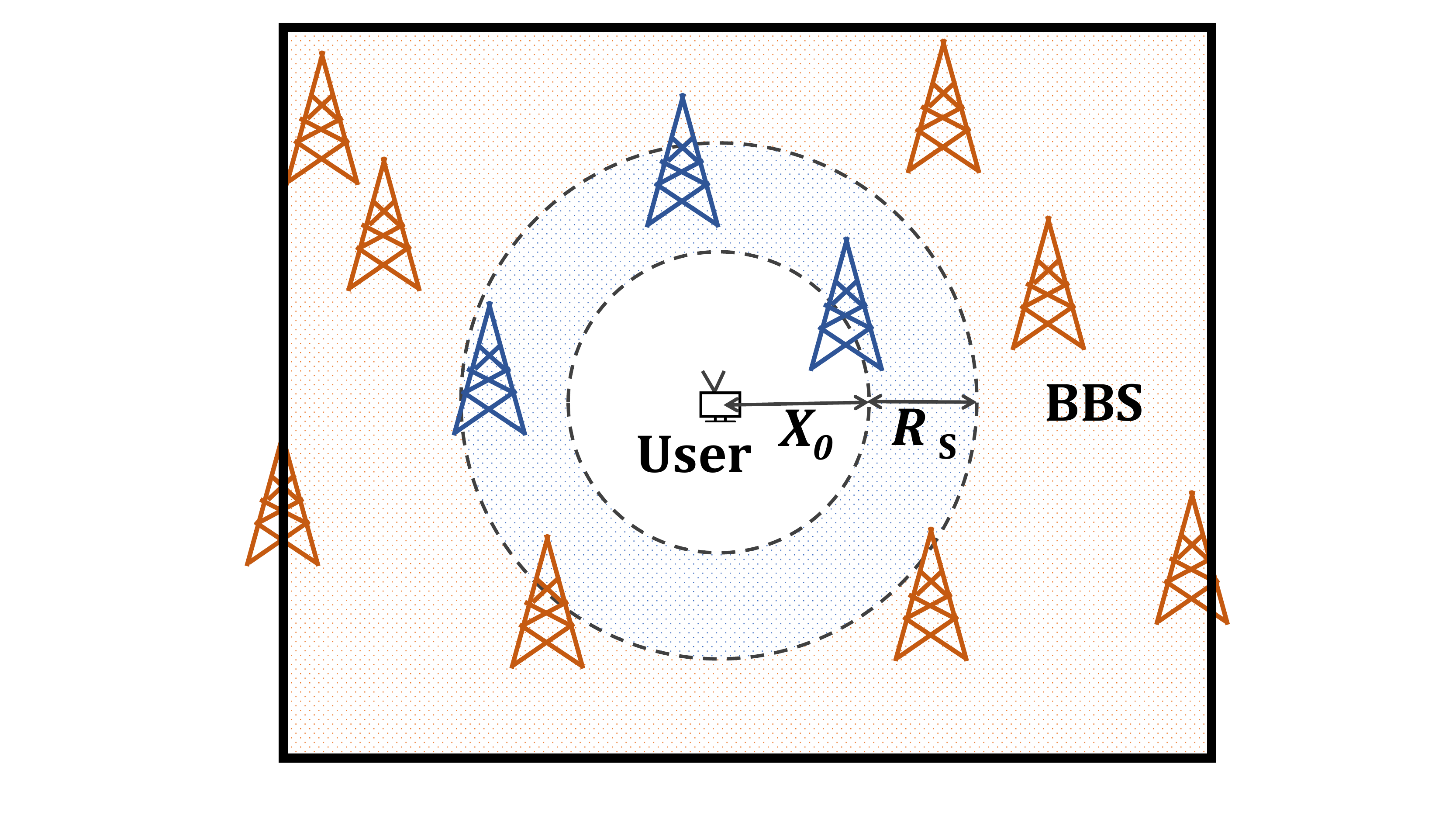}
   \caption{Illustration of system model of a broadcast network. A typical user is considered at the origin. ${X_0}$ is the distance of the nearest BS from the typical user. The 2D region $\Ball(\origin,{X_0}+R_\sym)$ denotes the connectivity region of the user.}
    \label{fig:System Modelfig}
\end{figure}

\subsection{Network Model}
The location of BBSs can be modeled as a homogeneous Poisson point process $\Phi=\{\X_{i}\in \mathbb{R}^2\}$ with density $\lambda$ in the region $\set{R}$ (See Fig. \ref{fig:System Modelfig}). Let $R_e=1/\sqrt{\lambda\pi}$ which represents the cell radius of an average cell. The subscribers (users) of the broadcasting service are assumed to form a stationary point process. We  assume a typical user located at the origin $\origin$. Consider each BS is operating in the same frequency band with the transmission bandwidth $W$. 
 Let $T_{\sym}$ is the symbol time of the transmitted symbol which is inversely proportional to the bandwidth $W$. Assume the transmit power of each BBS be $p_\tx$ and all devices are equipped with a single isotropic transmit antenna. The analysis can be extended for finite networks by taking $\set{R}=\Ball(\origin,R)$ with a finite $R$. 

 \subsection{Channel Model}
We assume the standard path-loss model. Hence, the received signal power from the $i$th BBS at the typical user at origin is given as
  \begin{align}
     P_i&= p_\tx \nf \beta_{i} {\dist{\X_{i}}}^{-\alpha},\label{eq:0.1}
  \end{align}
where   $X_i=\dist{\X_i}$ denotes the random distance of this BBS from the typical user.  Here, $\alpha$ is the path-loss exponent and $\nf$ is near-field gain which depends on the propagation environment. $\beta_{i}$ denotes the fading between the $i^{th}$ BBS and the user. We assume Rayleigh fading, \ie,    $\beta_{i}\sim\mathrm{Exp}(1)$ for tractability. 
  
  \subsection{Serving Signal and Interference Model}
In a broadcast system, multiple BBSs may transmit the same data at the one frequency band (as suggested by the name SFN). Therefore, at the receiver end, it can be seen as a single transmission with multi-path propagation and signals transmitted from multiple BBSs can be combined at the user.  However, since the signals from different BBSs are delayed according to time delays dependent on their distance, some of these signals may be delayed significantly and may overlap with the next transmission slots. Therefore, only those signals that have  delay within a certain limit can be combined to successfully decode the received symbol \cite{ahu21}.  The rest BBSs contribute to the  ISI (inter-symbol interference) which can be significant depending on the BBSs density. 

 Let $\X_0$ denotes the nearest serving BBS.  
 The probability density function of the distance    $X_0=\dist{\X_0}$ to the nearest BS from the user is given as\cite{Primer2016}
 \begin{align}
     f_{X_0}(\rcls)=2\pi\lambda \rcls e^{-\pi\lambda \rcls^2} \indside{\rcls\ge0}.\label{eq:5.10}
 \end{align}
 
 The time taken by the signal to reach from the $i^{th}$ BBS located at $\X_i$ to the typical user at $\origin$ be $T_i={X_i}/c$.  In particular, $T_0$ denotes the time taken by signal to reach from $\X_0$ to a typical user at $\origin$. 
Let the propagation delay of transmitted signal from $i^{th}$ BBS compared to the nearest serving BBS is $\DeltaT_i=T_i-T_0$.

 We assume that the receiver design allows the maximum delay of $\delta T_\sym$  for the received signals to be combined at the  user where $\delta\in[0\ 1]$ is a design parameter. This means that the received signal from the $i^{th}$ BBS  may contribute in the serving signal power if  $\DeltaT_i\leq \delta T_{\sym}$. This condition is equivalent to the condition  $\dist{\X_i}-\dist{\X_0}\leq R_\sym\stackrel{\Delta}{=} T_\sym\delta c$ on the BBSs location $\X_i$. In other words,  this means that all the BBSs that are located in the 2D region  $\{\X:\dist{\X}\le \dist{\X_0}+R_\sym \}=\Ball(\origin,{X_0}+R_\sym)$   can contribute to the serving signal  at the typical receiver at origin $\origin$. 
We term this region $\Ball(\origin,{X_0}+R_\sym)$ as 
  the connectivity region for the user and ${X_0}+R_\sym$ can be termed as the connectivity radius. 
Let $\msqrt^2={\lambda\pi R_\sym^2}$ denote the mean number of BBSs in this connectivity radius.

On the other hand, all the BBSs located outside $\Ball(\origin,{X_0}+R_\sym)$ \ie all the BBSs with ${X_i}\geq {X_0}+  R_\sym$ will contribute to the interference power even when they are transmitting the same data as their signal will be delayed beyond the specified limit.  

\subsection{Modeling Content Preferences of Users}
In this paper, we also include the fact that users may have their individual content or advertisement preferences. 
We assume that there  are $\Nc$ classes of users. Here, $\Nc$ is termed {\em content/advertisement granularity}. Each class of users prefers a particular type of content/ advertisements. We assume that the users will pay the network only when then can see a particular content of their interest. Each class consists of some quanta of users. For simplicity, we assume that each class has the same number of users, however, the presented framework can be trivially extended to include user classes with unequal sizes. We assume that one unit of revenue comes to the network from a particular class of users if every user of this class can see the content as per the preference of this class. Given the limited resources, the network can cater only to few classes and this capability depends on how these user classes are distributed spatially.  We will consider two types of users class distributions over the geographical space. 
We will also discuss a revenue model to characterize the network's revenue to help us understand optimal scheduling policies for the two scenarios.

\section{Coverage Analysis for Common Content Transmission}
We first start with the scenario that all users seek the same content, hence, all BBSs are transmitting the same content to everyone. Examples include systems transmitting emergency information, or traffic data which is common to every user. In this section, we will derive the  SINR and rate coverage probability for a typical user at the origin $\origin$ for such system. 

\subsection{SINR}
Since all BBSs are transmitting the same content, all BBSs located inside the connectivity region $\Ball(0,X_0+\R_\sym)$ contribute to the signal power. Therefore, the desired received signal power for the typical user at origin is given as
  \begin{align}
     \Prx'&=p_\tx \nf \beta_{0} {\dist{\X_{0}}}^{-\alpha}+\sum_{\X_{i}\in\Phi\cap \Ball(0,{X_0}+R_{\sym})\setminus \X_0} p_\tx \nf \beta_{i} {\dist{\X_{i}}}^{-\alpha}.\label{eq:1}
  \end{align}
 
 Similarly, the total interference  can be  given as 
   \begin{align}
I' &=\sum_{\X_{j}\in\Phi\cap \complementT{\Ball(0,\,X_0+R_{\sym})}}p_\tx \nf \beta_{j}{\dist{\X_{j}}}^{-\alpha}.\label{eq:2}
 \end{align}
 The signal-to-interference-plus-noise ratio (SINR) at the typical  receiver is given as
  \begin{align}
     \SINR&=\frac{\Prx'}{I'+N}
\iftoggle{SingleColumn}{}{\nonumber\\
  &}
=\frac{\beta_{0} {{X_{0}}}^{-\alpha}+\sum_{\X_{i}\in\Phi\cap \Ball(0,{X_0}+R_{\sym})\setminus\X_0} \beta_{i} {{X_{i}}}^{-\alpha}}
     {\sum_{\X_{j}\in\Phi\cap \complementT{\Ball(0,{X_0}+R_{\sym})}}\beta_{j}{{X_{j}}}^{-\alpha}+\sigma^{2}}.\label{eq:3}
 \end{align}
Here, $\sigma^2$ is the normalized noise power  given as $\sigma^2=N/(p_\tx \nf)$ where $N$ is the noise power. Similarly normalized  desired received signal power and interference are denoted by $S$ and $I$ which are given as $S=S'/(p_\tx \nf)$ and $I=I'/(p_\tx \nf)$. Hence, the SINR is  equal to
  \begin{align}
     \SINR&=\frac{\Prx}{I+\sigma^2}.
 \end{align}
Let $K=\sigma^2/R_e^{-\alpha}$ which represents the SNR at cell edge of an average cell.

\subsection{SINR Coverage Probability}\label{sinr_calc}
The SINR coverage probability $\Pc(\tau,\lambda)$ of a user is defined as the probability that the SINR at the user is above the threshold $\tau$ \ie
\begin{align}
     \Pc(\tau,\lambda)&= \prob{\SINR>\tau}
 \end{align}
 Using the conditioning on  the nearest serving BBS's location $\X_0$, the SINR coverage for typical user at $\origin$ is given as
\begin{align}
     \Pc(\tau,\lambda)&=\expects{\X_0}{\mathbb{P}\left(\SINR>\tau\right)\given}\nonumber\\&=\expects{\X_0}{\mathbb{P}\left(\frac{\Prx}{I+\sigma^{2}}>\tau\right)\given}
\iftoggle{SingleColumn}{}{\nonumber\\&}=\expects{\X_0}{\mathbb{P}\left(\Prx>(I+\sigma^{2})\tau\given\right)}\label{eq:3.2}.
 \end{align}
 Using the distribution of $\dist{\X_0}=X_0$, the SINR coverage probability can be further written as
\begin{align}
   \Pc(\tau, \lambda)  &=\expects{\X_{0}}{\mathbb{P}\left(\Prx>\tau\left(I+N\right)\,\vert\,\X_{0}\right)}\nonumber\\&=\int_{0}^\infty2\pi\lambda \rcls e^{-\pi\lambda \rcls^2}\mathbb{P}\left(\Prx>\tau\left(I+\sigma^2\right)\, \vert\,\dist{\X_{0}}=\rcls\right)\dd \rcls.\label{eq:3.4} 
\end{align}
To  solve the inner term further, we will use  Gil Pelaez's Lemma \cite{ahu11} which states that the CDF of a random variable $Y$ can be written in term of its Laplace transform $\laplace{Y}(t)$ as
\begin{align}
     F_Y(s)=\prob{Y \le s} = \frac{1}{2}- \frac{1}{\pi} \int_{0}^{\infty}\frac{1}{t}\Im{e^{-jts}\laplace{Y}(-jt)}\dd t.\label{eq:3.5}
 \end{align} 
 Using this Lemma, we get,
 \begin{align}
\mathbb{P}\left(\Prx>(I+\sigma^{2})\tau\given\right)
\iftoggle{SingleColumn}{}{\nonumber\\}&
=\expects{I\given}{
     \frac{1}{2}+ \frac{1}{\pi} \int_{0}^{\infty}\frac{1}{t}\Im{e^{-jt\tau(I+\sigma^2)}\laplace{\Prx}(-jt)}\dd t
     }\nonumber\\
     &=\frac{1}{2}+ \frac{1}{\pi} \int_{0}^{\infty}\frac{1}{t}\Im{\expects{I\given}{e^{-jt\tau(I+\sigma^2)}}\laplace{\Prx\given}(-jt)}\dd t
     \nonumber\\
 &=\frac{1}{2}+ \frac{1}{\pi} \int_{0}^{\infty}\frac{1}{t}\Im{\laplace{I\given}(jt\tau) e^{-jt\tau\sigma^2}\laplace{\Prx\given}(-jt)}\dd t. \label{eq:3.7}
 \end{align}
 Now, using  in \eqref{eq:3.4}, the SINR coverage probability is
 \begin{align}
   \Pc(\tau, \lambda) &=\frac{1}{2}+\frac{1}{\pi} \int_{0}^{\infty}\int_{0}^\infty2\pi\lambda \rcls e^{-\pi\lambda \rcls^2} \frac{1}{t}
\iftoggle{SingleColumn}{}{\nonumber\\&\ \ \ \times} \Im{\laplace{I\given}(jt\tau) e^{-jt\tau\sigma^2}\laplace{\Prx\given}(-jt)}\dd t\  \dd \rcls . \label{eq:cov_prob}
   \end{align}
  Here, $\laplace{I\given}(.)$ and $\laplace{\Prx\given}(.)$ are the Laplace transform of the sum interference $I$ and of the desired  received signal power $\Prx$ respectively which are given in the following Lemma.

  \begin{lemma}\label{lemma:1}
  The Laplace transforms of the desired signal  power and the sum interference  at the receiver located at origin $\origin$ are given as
  \begin{align}
    \mathcal{L}_{\Prx\,\vert\,\X_{0}}(s)&=\frac{1}{1+sX_0 ^{-\alpha}}\exp\left(-2\pi\lambda \int_{X_0}^{X_0+R_{\sym}}\frac{sr^{-\alpha}}{1+sr^{-\alpha}}r\,\dd r \right)\label{eq:3_10}\\
\laplace{I}(s)&=\exp{\left(-2\pi\lambda\int_{X_0+R_{\sym}}^{\RLimit}\frac{sr^{-\alpha}}{1+sr^{-\alpha}}r\,\dd r \right)}\label{eq:3.8}
\end{align}
%
  \end{lemma}
\begin{proof}
See Appendix \ref{proof1}.
 \end{proof}
 Using Lemma \ref{lemma:1} in \eqref{eq:cov_prob}, we can get the SINR coverage which is given in Theorem \ref{thm1}.

 \begin{theorem}\label{thm1}
 The probability of the SINR coverage for the user located at the origin in a broadcast network with $\lambda$ density of BBSs, is given as
\begin{align}
\Pc(\tau,\lambda)=&\frac{1}{2}+\frac{1}{\pi} \int_{0}^{\infty}2\pi\lambda \rcls e^{-\pi\lambda \rcls^2}\int_{0}^\infty\frac{1}{t}\Im{{\color{black}\frac{e^{-jt\tau\sigma^{2}}}{1-jt\rcls ^{-\alpha}}}\right.\nonumber\\
    &\hspace{20mm}\times\left.\exp\left(-2\pi\lambda\left(\int_{\rcls}^{\rcls+R_{\sym}}\frac{-jtr^{-\alpha}}{1-jtr^{-\alpha}}r\,\dd r
\iftoggle{SingleColumn}{}{\right.\right.\right.\nonumber\\&\hspace{30mm}\left.\left.\left.}
+\int_{\rcls+R_{\sym}}^{\RLimit}\frac{jt\tau r^{-\alpha}}{1+jt\tau r^{-\alpha}}r\,\dd r\right)\right)}\dd t\  \dd \rcls\nonumber
    \\
=&\frac{1}{2}+\frac{1}{\pi} \int_{0}^{\infty}\int_{0}^\infty
 	2v\frac1s
		\left[
		\frac{1}{1+s^2 v^{-2\alpha}}\right]
       e^{-v^2} e^{-s^{\frac2\alpha} M_d(s,v)}
		\nonumber\\
		&     \times   \left[
		s v^{-\alpha}
		\cos\left(s^{\frac2\alpha} N_d(s,v)+\tau s K\right)
        -\sin\left(s^{\frac2\alpha} N_d(s,v)+\tau s K\right)
		\right]
\dd v\;\dd s
\label{eq:3.26}
\end{align}
 where $M_d(t,\,\rcls)$ and $N_d(t,\,\rcls)$ are given as
\begin{align}
 M_d(s,v)&= \frac1\alpha\left[
		\bbfunc{
		    						\frac1\alpha,
		s^2(v+\msqrt)^{-2\alpha},
		s^2v.^{-2\alpha}
		}+
		       \tau ^{2/\alpha}
		\bbfunc{
		    						\frac1\alpha,
		0,
		\tau^2 s^2(v+\msqrt)^{-2\alpha}
		}\right]\label{thm1:mddef}
		\\   
 N_d(s,v)&= \frac1\alpha
		\left[
			-\bbfunc{
    						\frac1\alpha+\frac12,
			s^2(v+\msqrt)^{-2\alpha},
			s^2 v^{-2\alpha}
			}
			+\tau^{\frac2\alpha}
			\bbfunc{
    						\frac1\alpha+\frac12,
						0,
						\tau^2 s ^2 (v+\msqrt)^{-2\alpha}
					}
		\right]\label{thm1:nddef}\\
  						\text{with } &
\bbfunc{z,a,b}=\iftoggle{SingleColumn}{}{\nonumber\\&} \betafunc{\frac{1}{1+a}}{z}{-z+1}- \betafunc{\frac{1}{1+b}}{z}{-z+1}.
\end{align}

 \end{theorem}
 
 \begin{proof}
 See Appendix \ref{proof_thm1}.
 \end{proof}

 Theorem \ref{thm1} provides the SINR coverage in terms of two parameters: $K$ which denotes the inverse of SNR at the cell edge and $\msqrt^2$ which denotes the mean number of BBSs in connectivity radius circle. Further we can derive the following remarks.

\begin{remark}
For interference limited scenario, $K=0$, which means the coverage probability is a function of $\msqrt$ only. In case $\msqrt$ is fixed, individual variation of $\lambda$ and $R_\sym $ will not change the coverage.
\end{remark}
 
\begin{remark}\label{remark:sinrvariance}
For a broadcast network, an increase in the BBS density $\lambda$ improves both the desired signal power and the interference power. However, due to increase in number of serving BBSs due to increase in $\lambda$ which  improve the overall SINR coverage (which can also be seen in the numerical results). This behavior is different than conventional cellular case. Recall that with single serving BS density, the SINR in an interference-limited cellular network does not get affected by any increase in the BS density which is known as SINR invariance \cite{JGA2011}.  This can  be shown  from \eqref{eq:3.26} by performing a comparative study between $\lambda$ and $\lambda(1+\epsilon)$ with $\epsilon<1$ for some $\lambda$.
 \end{remark}

\begin{remark} \label{remark:impactRs}
It can been observed that the SINR coverage probability increases with an increases connectivity radius $R_\sym$ as it increases the serving power and decreases the interference. 
 \end{remark}

 \subsection{Rate Coverage Probability}
 The rate coverage probability of  a user is defined as the probability that the maximum achievable rate for the considered user is above some threshold $\rho$ \ie 
\begin{align*}
\Rc(\rho)&=\prob{\Rate>\rho}.
\end{align*}
Note that the maximum achievable rate for the typical user is given as 
\begin{align}
    \Rate&=\xi W\log_{2}(1+\SINR) \label{eq:RateShannon}
    \end{align}
  where $\xi$ is some coefficient that denotes the spectrum utilization. $W$ denotes the system bandwidth available to each BBS. 
   Hence, the rate coverage for the typical user is
   \begin{align}
    \Rc(\rho)&=\prob{\Rate>\rho}\nonumber\\
    &=\prob{\xi W\log_{2}(1+\SINR)>\rho}
    \nonumber\\
    &=\prob{\SINR>2^{\rho/(\xi W)}-1}=\Pc(2^{\rho/(\xi W)}-1)\label{eq:3.28}
    \end{align}
 where $\Pc$ is the SINR coverage probability given in \eqref{eq:3.26}.  Note that the available bandwidth $W$ affects $T_\sym$ and hence, $R_\sym$. If the BBSs use orthogonal frequency division multiplexing (OFDM) for transmission with FFT size $N_\sym$, then, $W$ is related to $T_\sym$ as
    $$W=\frac{N_\sym}{T_\sym}.$$
    Hence, the connectivity radius is 
    \begin{align}
    R_\sym=T_\sym\delta c=\frac{N_\sym\delta c}{W}.
    \label{eq:WTsymRelation1}
    \end{align}
    Hence, an increase in the system bandwidth increases the pre-log factor in \eqref{eq:RateShannon}, however, it also decreases the connectivity radius resulting in the lower SINR coverage probability. Therefore, we can observe a trade-off on the rate coverage with increasing bandwidth.
    
\subsection{Numerical Results}   
We now validate our results for SINR and rate coverage probabilities  through numerical simulation. We will also explore the impact of different parameters 
on the coverage probabilities via numerical evaluations of derived expressions to develop design insights. 
The default parameters  are given in Table \ref{table:Tabel1} which are according to \cite{ITUP1546,DBTdotPDF}.

\begin{table}[ht!]
\centering
\caption{Default parameters for numerical evaluations}
\label{table:Tabel1}
\begin{tabular}{|p{1in}|p{1.9in}|p{1in}|p{1.9in}|}
\hline
   \textbf{Parameters} & \textbf{Numerical value} & \textbf{Parameters} & \textbf{Numerical value} \\
    \hline
    $R_{\sym}$ & 19.18 km &
    $\lambda$ & 
    0.0014 
	BBSs / $\text{km}^2$\\
    \hline
    $N$&0 &
    Path-loss $a,\ \alpha$& $1.6\times10^{-3}, \ 4$\\
    \hline
   $W$ & 8 MHz&
    $p_\tx$ & 20 dB\\ \hline
    $N_{\sym}\delta$ &512 &
    Coefficients & $\nf=10^{-3},\ \xi=1$\\
    \hline $\Nc$ &15& 
    Simulation radius &  $800$ km\\
	\hline
\end{tabular}
\end{table}

%

\begin{figure}[ht!]
    \centering

    \includegraphics[width=\plotsize]{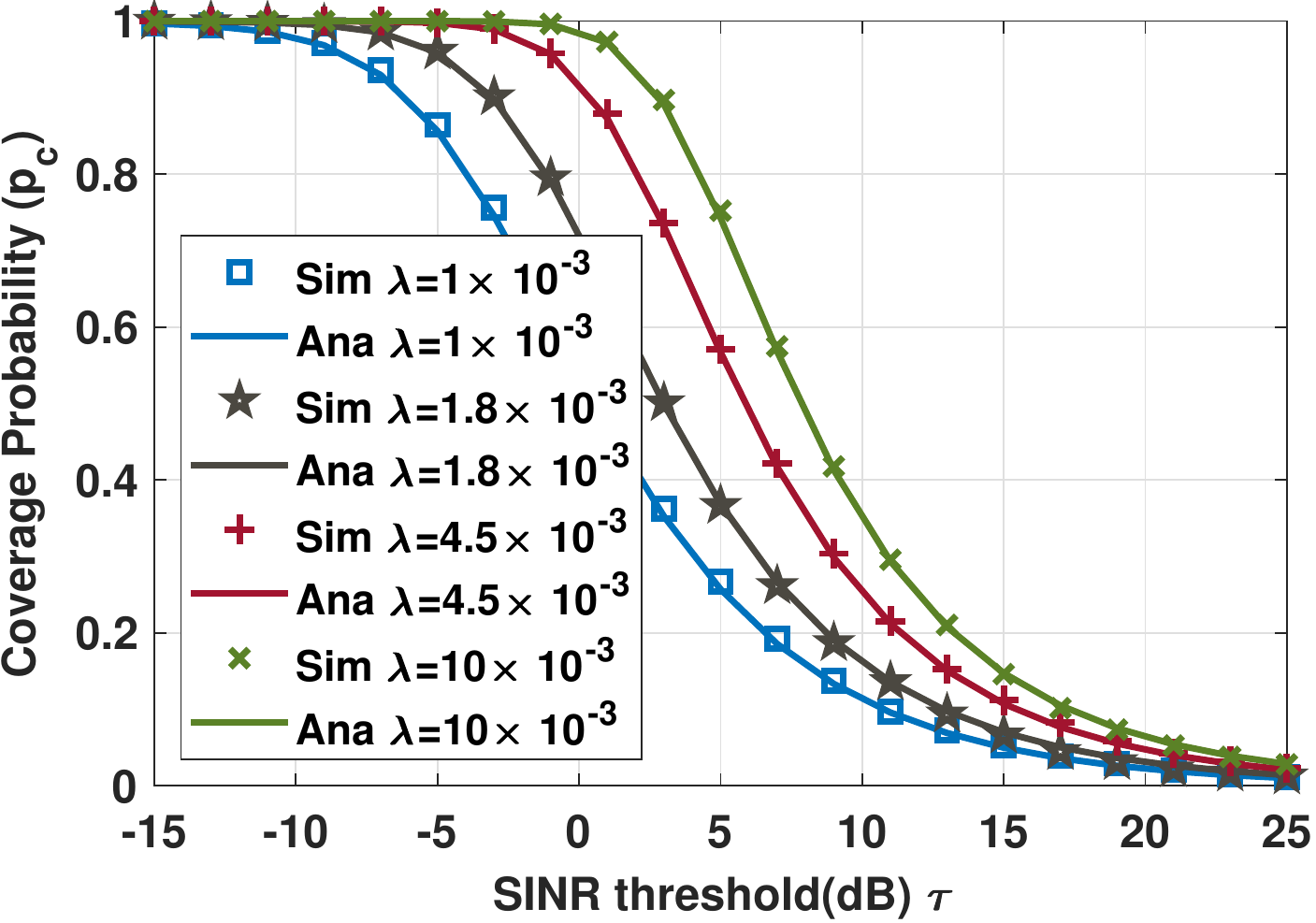}
   \caption{SINR coverage vs. SINR threshold ($\tau$) for various BBS density $\lambda$ in a broadcast system with multiple BBSs. Here, the solid lines represent the
analytical expression and markers represent simulation values. The parameters are according to Table \ref{table:Tabel1}. It can be seen that the analysis matches with simulation results. }
    \label{fig:4.1}
\end{figure}
\begin{figure}[ht!]
    \centering
    \includegraphics[width=\plotsize]{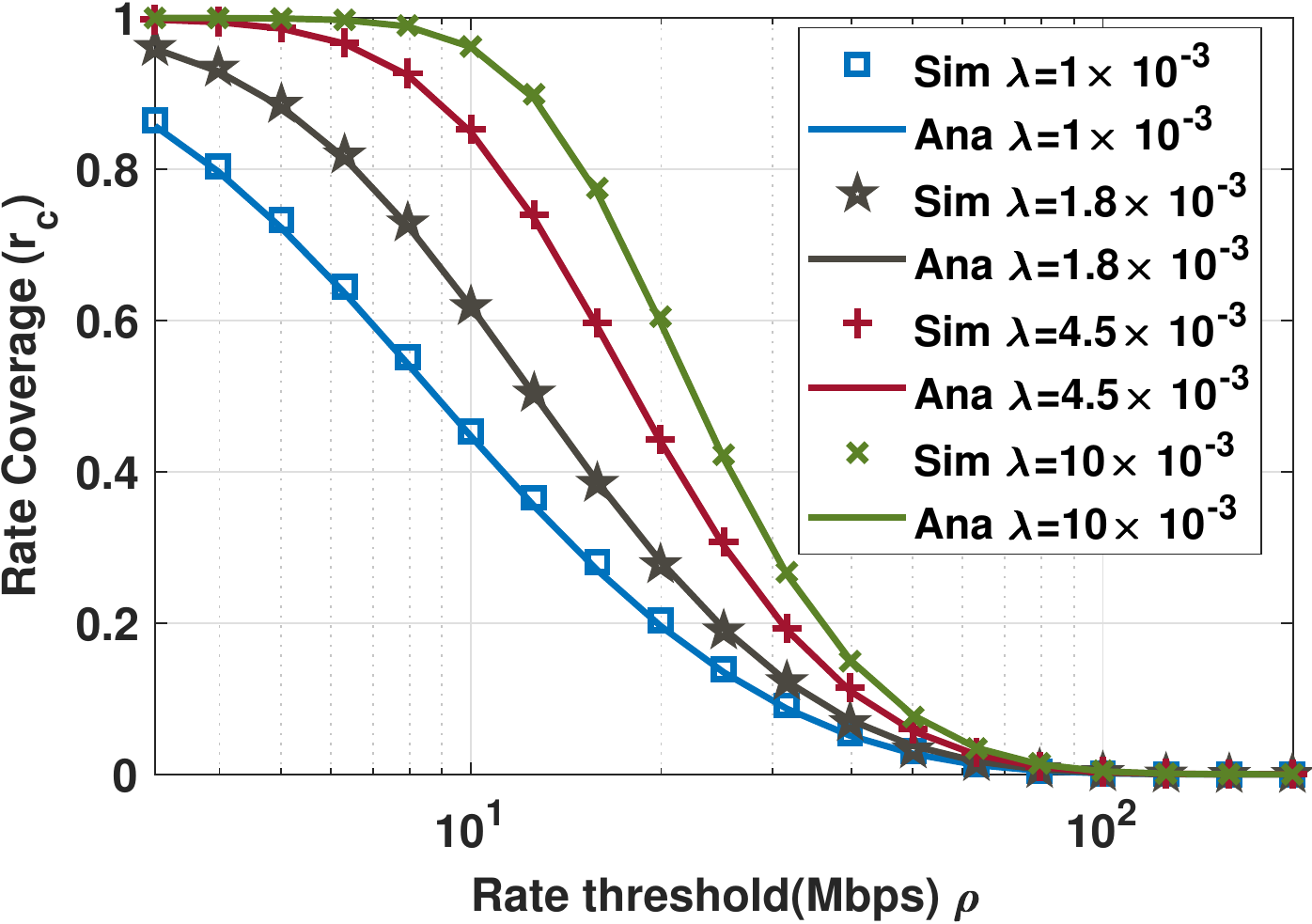}
   \caption{Rate coverage vs. rate threshold ($\rho$) for various BBS density $\lambda$ in a broadcast system with multiple BBSs. Here, the solid lines represent the
analytical expression and markers represent simulation values. The parameters are according to Table \ref{table:Tabel1}. It can be seen that the analysis matches with simulation results.
}
    \label{fig:4.2}
\end{figure}

\noindent\textbf{Validation of results:}
Fig. \ref{fig:4.1} shows the SINR coverage probability vs SINR threshold ($\tau$) for different values of BBSs density ($\lambda$). Here, the solid lines represent the analytical expression and markers represent simulation values. It can be seen that the analysis matches with simulation results which establishes the validity of the presented analytical results. From Fig. \ref{fig:4.1}, it can be seen that SINR coverage increases with an increase in the BBS density consistent with Remark \ref{remark:sinrvariance}. 
Similarly, Fig. \ref{fig:4.2} shows the rate coverage probability vs rate threshold ($\rho$) for different values of $\lambda$. It can be seen that the rate coverage increases with an increase in $\lambda$ which is expected due to the SINR coverage behavior with $\lambda$.

%

\begin{figure}[ht!]
    \centering

    \includegraphics[width=\plotsize]{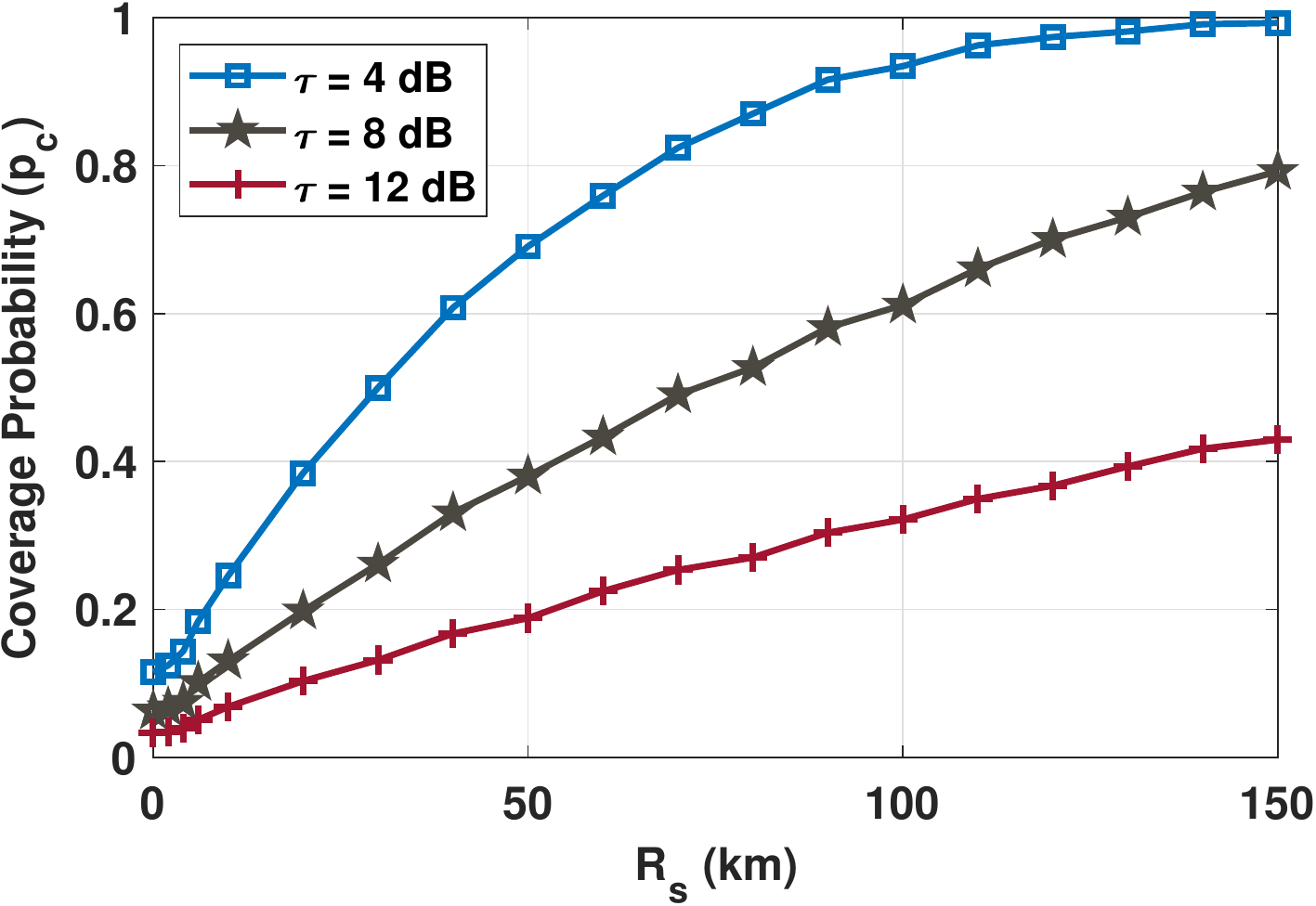}
   \caption{SINR coverage vs. $R_\sym$ for different values of SINR threshold ($\tau$) 
	in a broadcast network. Here, bandwidth varies with $R_\sym$ according to \eqref{eq:WTsymRelation1} with maximum value at 80 MHz. The rest of the parameters are according to 
	Table \ref{table:Tabel1}.  It is observed that the SINR improves with $R_\sym$.
 }
    \label{fig:4.4}
\end{figure}

\begin{figure}[ht!]
    \centering

    \includegraphics[width=\plotsize]{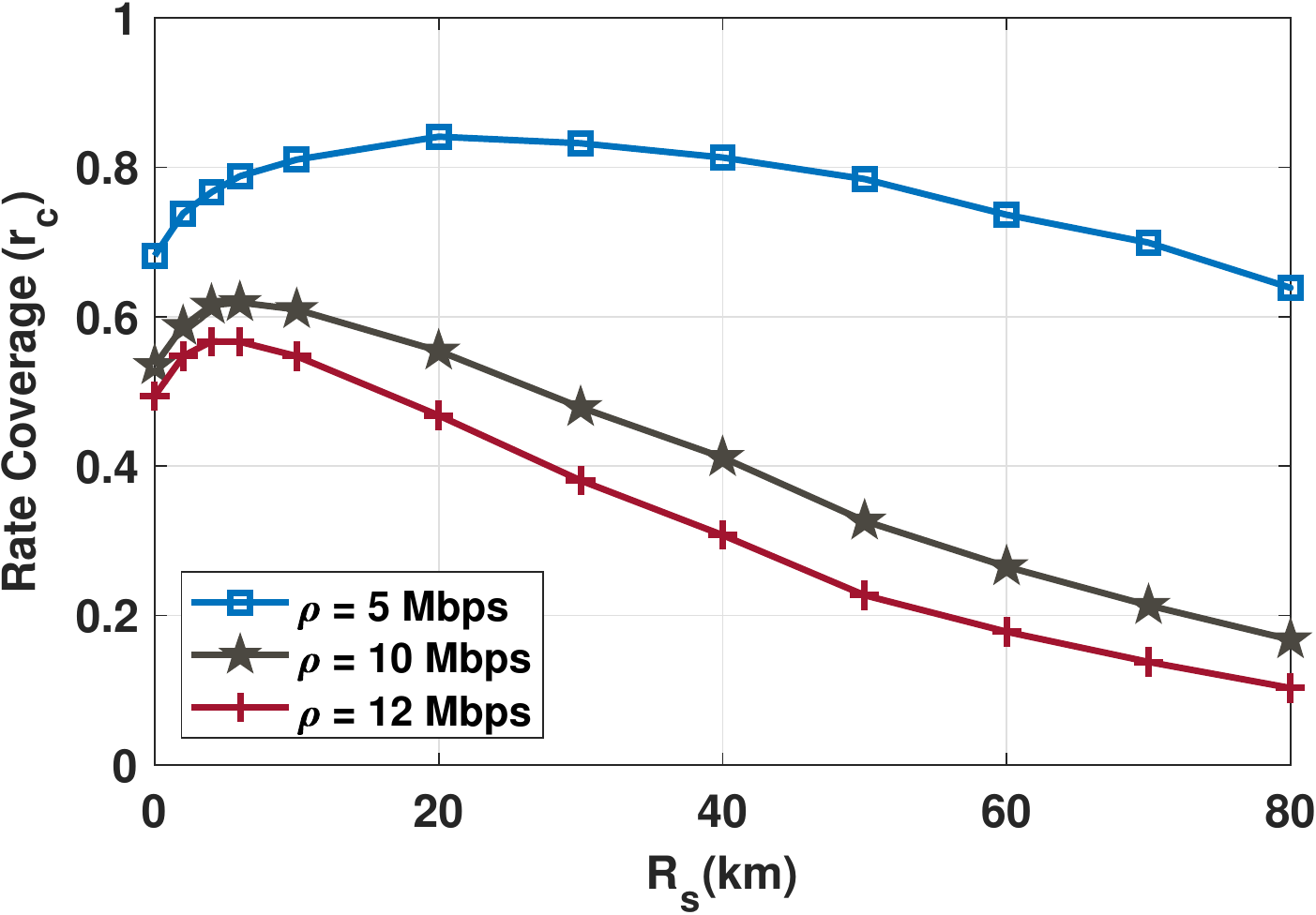}
   \caption{ 
   	Rate coverage vs. $R_\sym$ for different values of rate threshold ($\rho$ in Mbps) 
	in a broadcast network. Here, bandwidth varies with $R_\sym$ according to \eqref{eq:WTsymRelation1} with maximum value at 50 MHz. The rest of the parameters are according to 
	Table \ref{table:Tabel1}.  A trade-off is seen in rate the with
varying $R_\sym$.
   }
    \label{fig:4.5}
\end{figure}

\noindent\textbf{Impact of connectivity radius on SINR and rate coverage:} 
Fig. \ref{fig:4.4} shows the variation of SINR coverage with  the connectivity radius ($R_{\sym}$) for different values of target SINR threshold. It is observed that the SINR coverage increases with $R_\sym$. It can be justified as $R_\sym$ increases the number of contributing BBSs increases and the number of interfering BBSs decreases. 

 Fig.  \ref{fig:4.5} shows the variation of the rate coverage with $R_\sym$. We can observe that with $R_\sym$,  the rate coverage first 
increases up to a certain value of $R_\sym$ and afterwards starts decreasing again. From \eqref{eq:WTsymRelation1}, increase in $R_\sym$ requires a decrease in the bandwidth $W$ in order to allow larger symbol time. This causes a trade-off in the system performance. As bandwidth is a pre-log factor in the rate expression, it has a negative and larger impact on the rate coverage. Hence, beyond a certain value of $R_\sym$, the impact of decrease in $W$ dominates the increase in SINR caused by increased $R_\sym$ which results in the decrease in the rate.  Due to the same reasons, there may exist an optimal $R_\sym$ that maximizes the rate coverage. The knowledge of optimal $R_\sym$ can be helpful in designing the broadcast network.

\noindent\textbf{Impact of network density on SINR coverage:} 
Fig. \ref{fig:4.6}(a) and (b) show the variation of SINR and rate coverage with  the network density $\lambda$ for different value of target SINR threshold and connectivity radius (which is achieved by changing bandwidth while keeping other parameters the same as Table I). The coverage while ignoring the noise is also shown. From Fig. \ref{fig:4.5}(a), it can be seen that densification of the network helps in SINR coverage. When the BBS density is very small, network is noise limited. As $\lambda$ increases, BBSs comes closer to the user  improving SINR coverage while SIR coverage remains constant. After a certain $\lambda$, the increase in $\lambda$ reduces the interference also. Hence, both SIR and SINR coverage improves. At high value of $\lambda$, coverage becomes 1 as all dominant BBSs provide serving power. The behavior of SIR coverage is similar to as seen in networks with dual-slope pathloss \cite{GupZhaAnd2015}.

\begin{figure}[ht!]
    \centering
    \includegraphics[width=8cm]{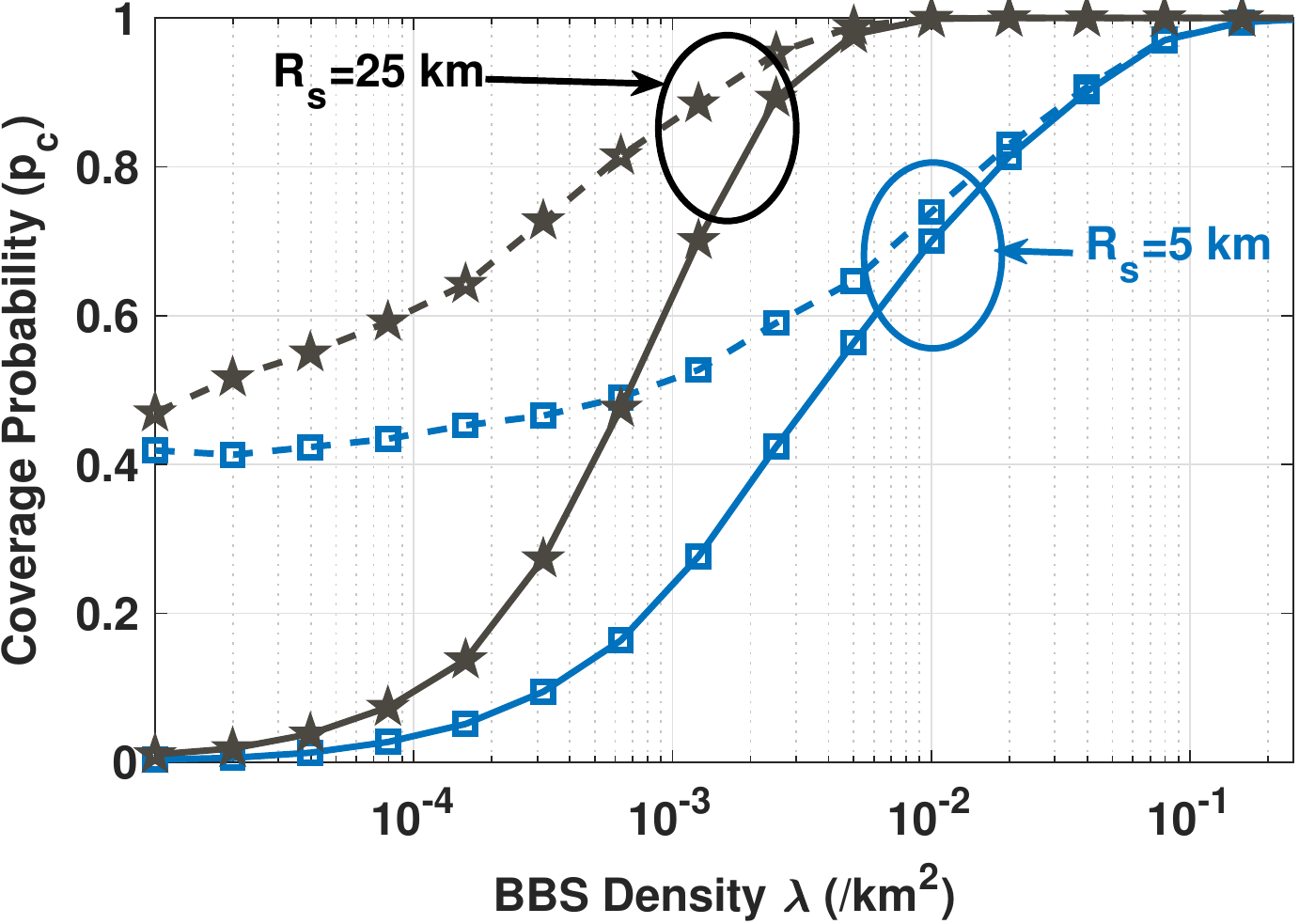}\ 
    \includegraphics[width=8cm]{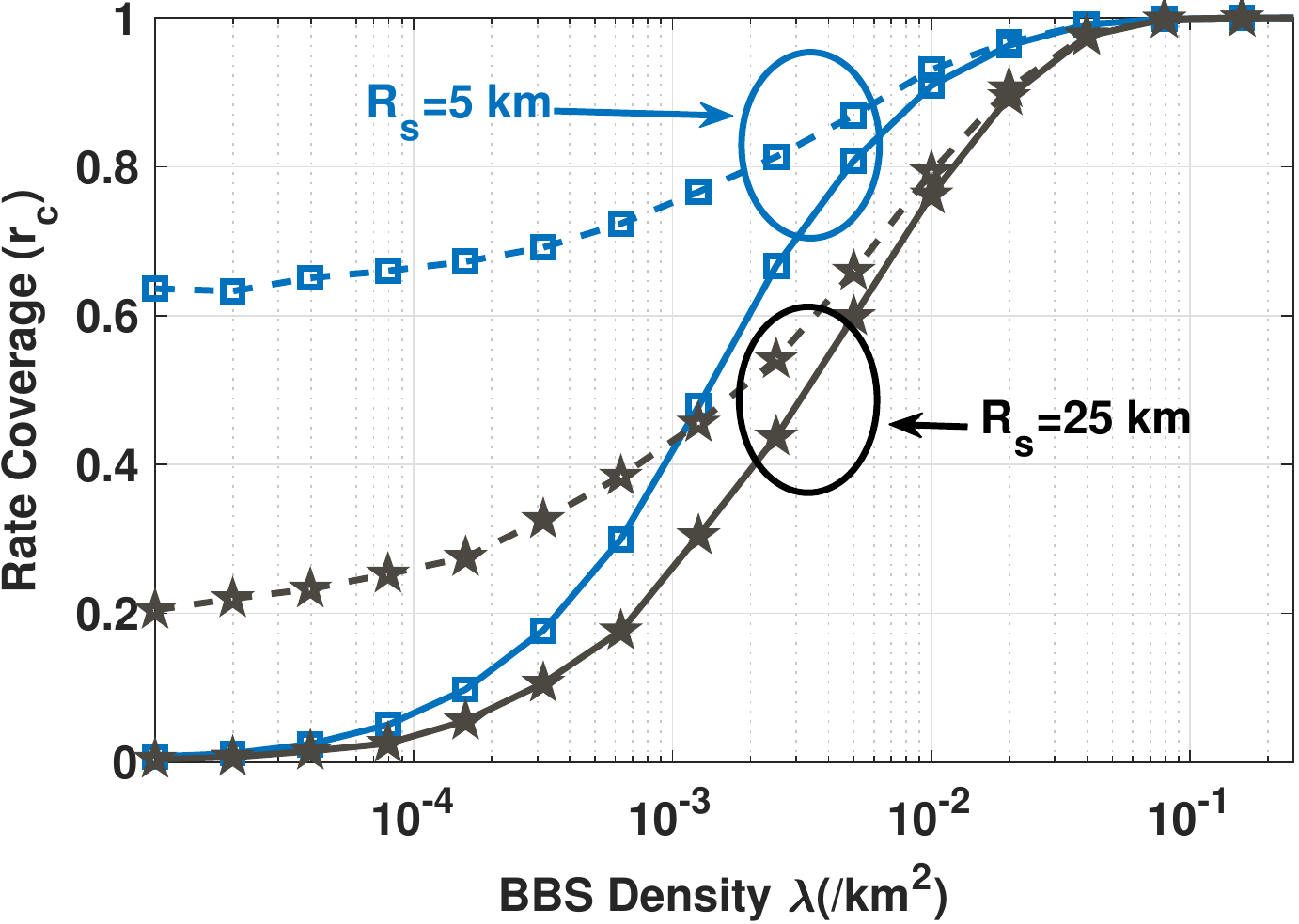}
   \caption{ 
   	Coverage vs. BBS density $\lambda$ for different values of $R_\sym$
	in a broadcast network. Here, rest of the parameters are according to 
	Table \ref{table:Tabel1}. (a)  SINR coverage. Dashed lines indicate SIR coverage while ignoring noise. 
 (b)  Rate coverage. Dashed lines indicate rate coverage  ignoring the  noise. 
   }
    \label{fig:4.6}
\end{figure}

 \section{Scenario I: Coverage Analysis for Networks with  Users  having a High Level of Spatial Heterogeneity in Content Preference}\label{sec:sc:1}

We now extend the system model to include networks with users having their individual content or advertisement preferences. 
In this section, we consider the first scenario where there is high level of spatial heterogeneity in users. This means that all classes of users are present in any region. Given the limitation of resources, network selects $n$ classes of users and shows  $n$ contents (one for each class) at  any point of time. Here, $n$ is a design parameter decided by the broadcast network. Since the user classes are spatially inseparable, each BBSs should transmit to the same $n$ contents to improve SINR coverage.  We have assumed OFDM based transmission where a BBS transmits the $n$ number of contents on orthogonal resources.  

\subsection{SINR and SINR Coverage}
To improve coverage, the network can use the same bands for a particular content across all BBSs. Since all the BBSs are transmitting the same data in a band, the SINR for a typical user is the same as given in \eqref{eq:3}. Similarly, in this case, the SINR coverage probability of a typical user  is the same as given in Theorem \ref{thm1}.

\subsection{Rate Coverage}
Now, the available resources are divided among $n$ contents. If the total available bandwidth is $W$,  the bandwidth available for each content is $W/n$. The instantaneous achievable rate for a typical user located at origin, for each content is 
\begin{align}
    \Rate&=\xi\frac{W}{n}\log_{2}(1+\SINR).\label{eq:5.1}
\end{align}
From \eqref{eq:3.28}, the rate coverage probability is given as:
\begin{align}
    \Rc(\rho)
    &=\prob{\SINR>2^{n\rho/ (\xi W)}-1}=\Pc(2^{n\rho/ (\xi W)}-1)\label{eq:5.3}
    \end{align}
 where $\Pc$ is given in \eqref{eq:3.26}.

\subsection{Network Revenue}
Let $\rho$ denote the minimum rate required for a user to be able to view the content. Then, the rate coverage $\Rc$ at  $\rho$ denotes the  fraction of users that are able to view this content. Therefore, $\mathrm{r}_c$ unit of revenue will be earned by the network from a particular class, since only $\Rc$ fraction of users can watch it.
 Therefore, the network's total revenue $\Rn$ can be given as:
\begin{align}
    \Rn&=n\mathrm{r}_{c}(\rho).
    \label{eq:5.2}
\end{align}

\begin{figure}[ht!]
    \centering

    \includegraphics[width=\plotsize]{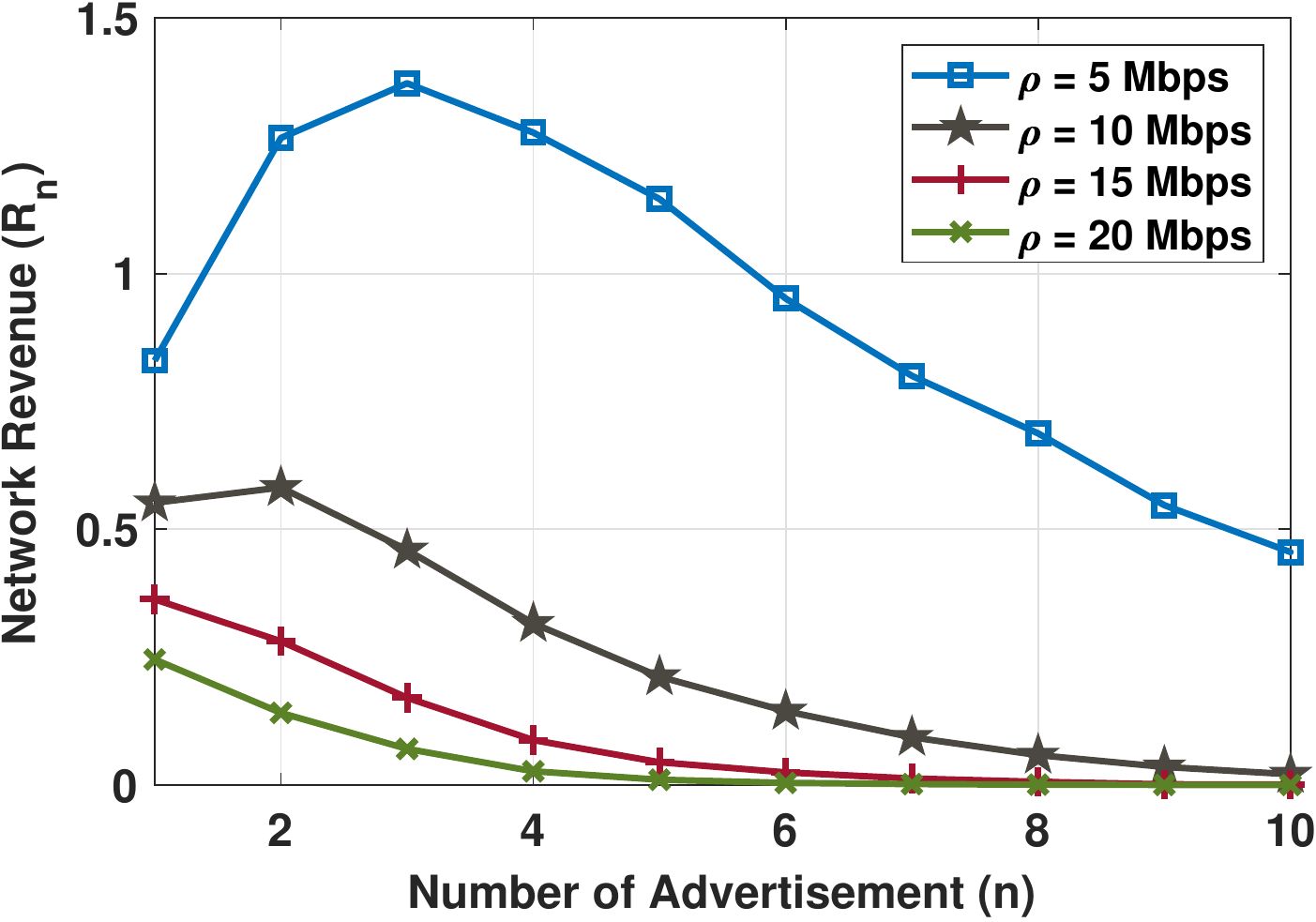}\ 
   \caption{Variation of total revenue with respect to allowed number of user classes for different rate threshold $\rho$ (Mbps) (which is a proxy for content quality requested). Content granularity is  $\Nc=15$. Other parameters are according to Table \ref{table:Tabel1}. It is observed that an optimal value of $n$ can provide the maximum revenue to the network which depends on the content quality requested. }
    \label{fig:5.1}
\end{figure}
\subsection{Numerical Results: Impact of $n$ on Total Revenue}

Fig.  \ref{fig:5.1} shows the variation of the network revenue $\Rn$ with $n$ for a system with $\Nc=15$. We can observe that initially, the revenue increases with an increase in $n$ up to a certain value, and thereafter, starts decreasing. This  can be justified in the following way. If $n$ increases, the following two effects take place-- (1) more user classes are served, causing a linear increase in $\Rn$, and (2) available bandwidth $W/n$ for each content/advertisement decreases causing the rate coverage to drop. As a combined effect dictated by \eqref{eq:5.2}, the revenue is optimal at a particular value of $n$. 
However, this behavior also depends on the target rate threshold. At a higher rate threshold, $\Rn$ decreases with $n$, showing $n=1$ as the optimal choice. This implies that the optimal number of user classes that can be served, depends on the quality of the content. If the quality requested is high, then it may be better to serve less user classes, while more user classes can be served when the quality requested is lower.

 
 \section{Scenario II: Coverage Analysis for Networks with  Users  having  Spatially Separated Classes for Content Preference}\label{sec:sc:2}
The second scenario we consider corresponds to the case where  user classes are spatially separated. 
For tractability, we assume that coverage area of each BBS comprises of users from a single class. Thus, users of different classes are spatially separated. Let $S_i$ denote the class of the $i$th BBS located at $\X_i$ which is a uniform discrete random variable  with PMF given as
\begin{align}
p_{S_i}(k)&=\frac1{\Nc}\indside{1\le k\le\Nc}.
\end{align}
We assume that $S_i$'s are independent of each other.
%
%

The network can cater to all user classes by letting each BBS to transmit content according to the class of users lying in its coverage area. Note that for a typical user, only those BBSs that transmit the same content, can contribute to the serving signal power at this user. Therefore, this strategy will reduce the number of serving BSs and hence reduce the SINR. On the other hand, network can decide to cater only one user class by forcing all BBSs to transmit only one content, will reduce the revenue as users of only one class will receive their preferred content. It will be interesting to find the optimal number of user-classes that can be catered by the network. 
As a general problem, we consider that the network decides to cater $n$ user classes out of total $\Nc$ classes. Let us denote the set of all selected index by $\selset$.

Let us denote the content transmitted by $i$th BS by $M_i$. If the BBS's user class is one of the $n$ selected classes (\ie $S_i\in \selset$),  it will transmit the content corresponding to its user class \ie $M_i=S_i$. If the BBS's user class is not one of the selected classes, it will transmit the content corresponding to a randomly selected user class to help boost its signal strength.


Let us consider a typical user at $\origin$. Without loss of generality, assume that its user class is 1. The probability that its class is one of the $n$ selected classes is $n/\Nc$. Let us condition on the fact that it is one of the selected classes. 

 Let the tagged BS of this typical user transmits the content $M_0=1$. Then, for the $i$th BBS, $M_i=M_0$ if
 \begin{enumerate}
 \item $S_i=M_0$ which occurs with probability $\frac1{\Nc}$, or
 \item $S_i \notin \selset$ and $S_i=M_i$ which occurs with probability $\frac{\Nc-n}{\Nc}\cdot \frac1{n}$.
 \end{enumerate}
 %
  Therefore, the probability that the $i$th BBS is transmitting the content as per the preference of this typical user is 
\begin{align}
    p&=\prob{M_i=M_0}=\frac{1}{\Nc}+\frac{\Nc-n}{\Nc}\cdot\frac{1}{n}=\frac{1}{n}.\label{eq:sameprob}
    \end{align}

\subsection{SINR}
Now, note that the BBSs that are transmitting the same content as $M_0$ and are located inside $\Ball(\origin,{X_0}+R_\sym)$ will contribute to the desired signal power to the typical user at origin $\origin$. 
Therefore, the desired signal power for the typical user at origin is given as
 \begin{align}
     \Prx=&p_\tx \nf \beta_{0} {{X_{0}}}^{-\alpha}+
	 \iftoggle{SingleColumn}{}{\nonumber\\
	&}
	\sum_{\X_{i}\in\Phi\cap \Ball(0,X_0+R_{\sym})\setminus \X_0} p_\tx \nf \beta_{i} {{X_{i}}}^{-\alpha}\indside{M_i=M_0}.\label{eq:5.4}
  \end{align}
Similarly, the interference for the typical user is caused by the BBSs that are either located outside $\Ball(\origin,{\rcls}+R_\sym)$ or located inside the $\Ball(\origin,{X_0}+R_\sym)$ but transmitting a different content than $M_0$. Hence, the total interference is given as 
\begin{align}
I =&\sum_{\X_{i}\in\Phi\cap \Ball(0,X_0+R_{\sym})\setminus \X_0} p_\tx \nf \beta_{i}{{X_{i}}}^{-\alpha}\indside{M_i\ne M_0}
 \iftoggle{SingleColumn}{}{\nonumber\\&}+\sum_{\X_{j}\in\Phi\cap \complementT{\Ball(0,{X_0}+R_{\sym})}}p_\tx \nf \beta_{j}{{X_{j}}}^{-\alpha}.\label{eq:5.5}
 \end{align}
 Now, the SINR for this user is given as 
   \begin{align}
     &\SINR=\frac{\Prx}{I+N}\nonumber\\
     &=\frac{ \beta_{0} {\dist{\X_{0}}}^{-\alpha}+\sum_{\X_{i}\in\Phi\cap \Ball(0,\rcls+R_{\sym})\setminus \X_0} \beta_{i} {\dist{\X_{i}}}^{-\alpha}\indside{M_i=M_0}}{\sum_{\X_{i}\in\Phi\cap \Ball(0,\rcls+R_{\sym})\setminus \X_0}\beta_{i}{\dist{\X_{i}}}^{-\alpha}\indside{M_i\ne M_0}+\sum_{\X_{j}\in\Phi\cap \complementT{\Ball(0,\rcls+R_{\sym})}} \beta_{j}{\dist{\X_{j}}}^{-\alpha}+\sigma^2}.
    \end{align}
    Here, $\sigma^2$ is the normalized noise power which is given as $\sigma^2=N/(p_\tx \nf)$ where $N$ is the noise power.
    
\subsection{SINR Coverage Probability}
We now calculate the SINR coverage for a typical user. 
Similar to Section \ref{sinr_calc}, the SINR coverage probability is given as
 \begin{align}
   \Pc(\tau, \lambda) &=\frac{1}{2}+\frac{1}{\pi} \int_{0}^{\infty}\int_{0}^\infty2\pi\lambda \rcls e^{-\pi\lambda \rcls^2} \frac{1}{t}
   \iftoggle{SingleColumn}{}{\nonumber\\&\ \ \ \times} \Im{\laplace{I\given}(jt\tau) e^{-jt\tau\sigma^2}\laplace{\Prx\given}(-jt)}\dd t\  \dd \rcls  \label{eq:cov_prob_2}
   \end{align}
   where $\Prx$ and $I$ are given in \eqref{eq:5.4} and \eqref{eq:5.5}.
\begin{lemma}\label{lemma2}
Conditioned on the location of the closest serving BBS, the Laplace transforms of the desired signal  power and the sum interference  at the receiver are given as
\begin{align}
&\mathcal{L}_{\Prx\given}(s)
 \iftoggle{SingleColumn}{}{\nonumber\\
    &}
    =\frac{1}{1+s\rcls ^{-\alpha}}\exp\left(-2\pi\lambda p \int_{X_0}^{X_0+R_{\sym}}\frac{sr^{-\alpha}}{1+sr^{-\alpha}}r\dd r \right)\label{eq:5.22}\\
&\laplace{I|\X_0}(s) 
=\exp\left(-2\pi\lambda(1-p)\int_{X_0}^{X_0+R_{\sym}}\frac{sr^{-\alpha}}{1+sr^{-\alpha}}r\dd r
 \iftoggle{SingleColumn}{}{\right.\nonumber\\&\hspace{25mm}\left.}
 -2\pi\lambda\int_{X_0+R_{\sym}}^{\RLimit}\frac{sr^{-\alpha}}{1+sr^{-\alpha}}r\dd r \right).\label{eq:5.17}
    \end{align}

%
\end{lemma}
\begin{proof}
See Appendix \ref{proof2}.
\end{proof}
Using Lemma \ref{lemma2} and \eqref{eq:cov_prob_2} we can calculate the SINR coverage  which is given in Theorem \ref{thm2}.
\begin{theorem}\label{thm2}
  The probability of SINR coverage for the user located at the origin in a broadcast network with $\lambda$ density of BBSs, is given as:
\begin{align}
\Pc(\tau,\lambda)=&
	\frac{1}{2}+\frac{1}{\pi} \int_{0}^{\infty}2\pi\lambda \rcls e^{-\pi\lambda \rcls^2}\int_{0}^\infty\frac{1}{t}
\Im{
{\color{black}\frac{e^{-jt\tau\sigma^{2}}}{1-jt\rcls ^{-\alpha}}}
\right.\nonumber\\
&\indent\times \left.\exp
\left(-2\pi\lambda\left( p\int_{\rcls}^{\rcls+R_{\sym}}\frac{-jtr^{-\alpha}}{1-jtr^{-\alpha}}r\,\dd r
\iftoggle{SingleColumn}{}{\right.\right.\right.\nonumber\\&\hspace{10mm}\left.\left.\left.}
+(1-p)\int_{\rcls}^{\rcls+R_{\sym}}\frac{jt\tau r^{-\alpha}}{1+jt\tau r^{-\alpha}}r\,\dd r
\right.\right.\right.\nonumber\\&\hspace{15mm}\left.\left.\left.
+\int_{\rcls+R_{\sym}}^{\infty}\frac{jt\tau r^{-\alpha}}{1+jt\tau r^{-\alpha}}r\,\dd r\right)\right)}
\dd t \dd \rcls\nonumber \\
=&\frac{1}{2}+\frac{1}{\pi} \int_{0}^{\infty}\int_{0}^\infty
 	2v\frac1s
		\left[
		\frac{1}{1+s^2 v^{-2\alpha}}\right]
       e^{-v^2} e^{-s^{\frac2\alpha} M'_d(s,v)}
		\nonumber\\
		&     \times   \left[
		s v^{-\alpha}
		\cos\left(s^{\frac2\alpha} N'_d(s,v)+\tau s K\right)
        -\sin\left(s^{\frac2\alpha} N'_d(s,v)+\tau s K\right)
		\right]
\dd v\;\dd s
\label{eq:5.33} 
\end{align}
%
%
 where $M'_d(t,\,\rcls)$ and $N'_d(t,\,\rcls)$ is given as
\begin{align}
 M_d(s,v)=& \frac1\alpha\left[
		p\bbfunc{
		    						\frac1\alpha,
		s^2(v+\msqrt)^{-2\alpha},
		s^2v^{-2\alpha}
		}+
		       (1-p)\tau ^{2/\alpha}
		\bbfunc{
		    						\frac1\alpha,
		\tau^2s^2(v+\msqrt)^{-2\alpha},
		\tau^2s^2v^{-2\alpha}
		}\right.\nonumber\\&\left.
		+
		       \tau ^{2/\alpha}
		\bbfunc{
		    						\frac1\alpha,
		0,
		\tau^2 s^2(v+\msqrt)^{-2\alpha}
		}\right]\label{thm2:mddef}
		\\   
 N_d(s,v)=&\frac1\alpha\left[
		-p\bbfunc{
		    						\frac1\alpha+\frac12,
		s^2(v+\msqrt)^{-2\alpha},
		s^2v^{-2\alpha}
		}\right.\nonumber\\
		&\left.+
		       (1-p)\tau ^{2/\alpha}
		\bbfunc{
		    						\frac1\alpha+\frac12,
		\tau^2s^2(v+\msqrt)^{-2\alpha},
		\tau^2s^2v^{-2\alpha}
		}\right.\nonumber\\&\left.
		+
		       \tau ^{2/\alpha}
		\bbfunc{
		    						\frac1\alpha+\frac12,
		0,
		\tau^2 s^2(v+\msqrt)^{-2\alpha}
		}\right]\label{thm2:nddef}
\end{align}
\end{theorem}
\begin{proof}
See Appendix  \ref{proof_thm2}.
\end{proof}

\subsection{Rate Coverage}
Since each BBS shows only one content, it can use the total available bandwidth $W$. The instantaneous achievable rate for a typical user located at origin, while receiving the content, is 
\begin{align}
    \Rate&=\xi{W}\log_{2}(1+\SINR).
\end{align}
From \eqref{eq:3.28}, the rate coverage probability is given as:
\begin{align}
    \Rc(\rho)
    &=\prob{\SINR>2^{\rho/ (\xi W)}-1}=\Pc(2^{\rho/ (\xi W)}-1)\label{eq:3.28_2}
    \end{align}
 where $\Pc$ is given in Theorem \ref{thm2}.

\subsection{{Network Revenue}}

Let $\rho$ denote the minimum rate required for a user to be able to view the content. Then, the rate coverage $\Rc$ at  $\rho$ denotes the  fraction of users that are able to view this content. Therefore, $\mathrm{r}_c$ unit of revenue will be earned by the network from a particular class, since only $\Rc$ fraction of users can watch it.
   The probability that the typical user receives the content as per its preference is $n/\Nc$ which is also the probability that the network will receive revenue from this typical user.  
Similar to previous sections, the total revenue can be computed as
\begin{align}
\Rn&=\frac{n}{\Nc}\cdot\rcv(\rho),\label{eq:5.37}
\end{align}
where the rate coverage $\Rc$ for the typical user is given by \eqref{eq:3.28_2}.

\subsection{Numerical Results}
We now present numerical results for the considered scenario II. The parameters are stated in Table \ref{table:Tabel1}.
%
 


 \begin{figure}[ht!]
    \centering

    \includegraphics[width=8cm]{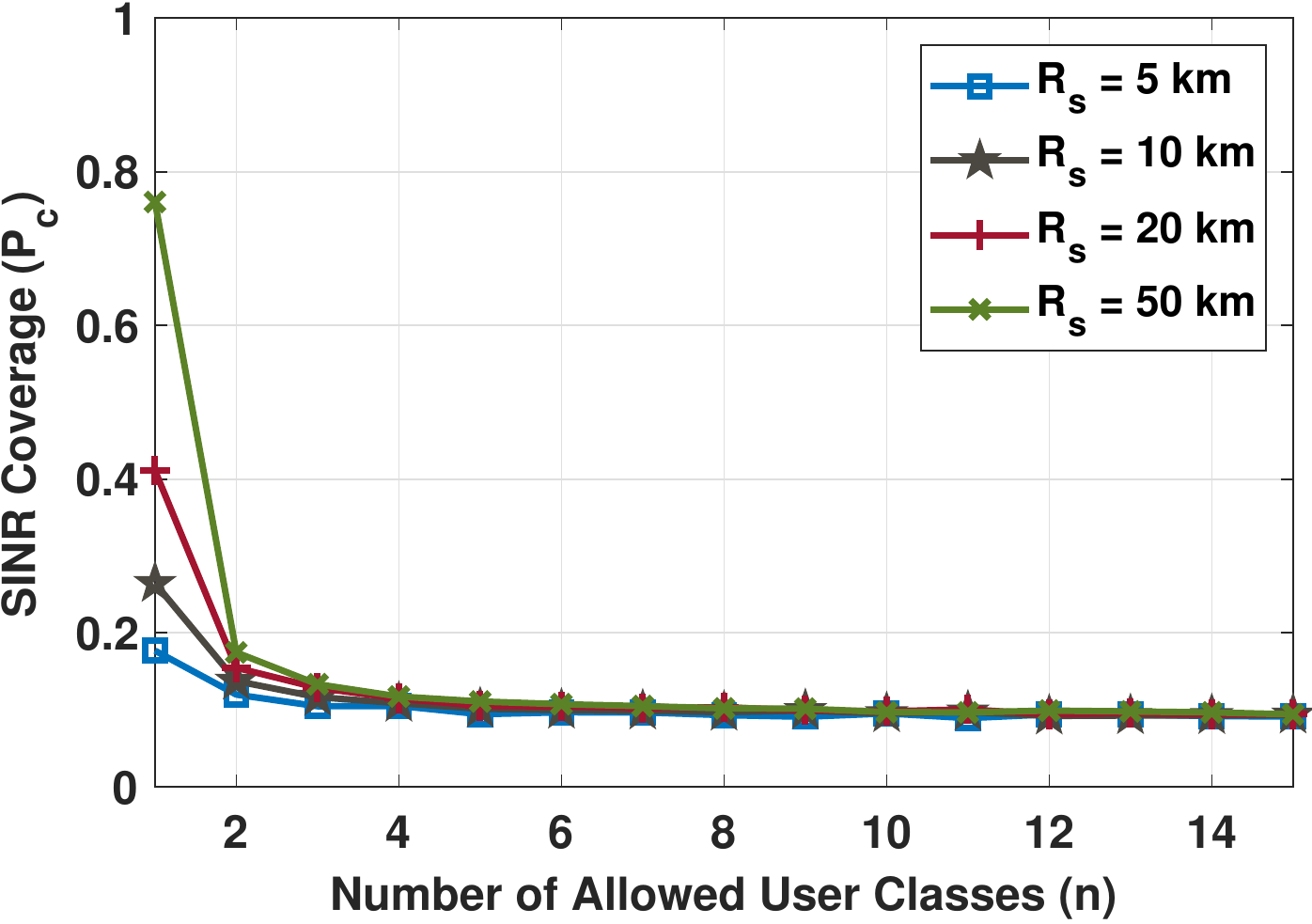}\  
    \includegraphics[width=8cm]{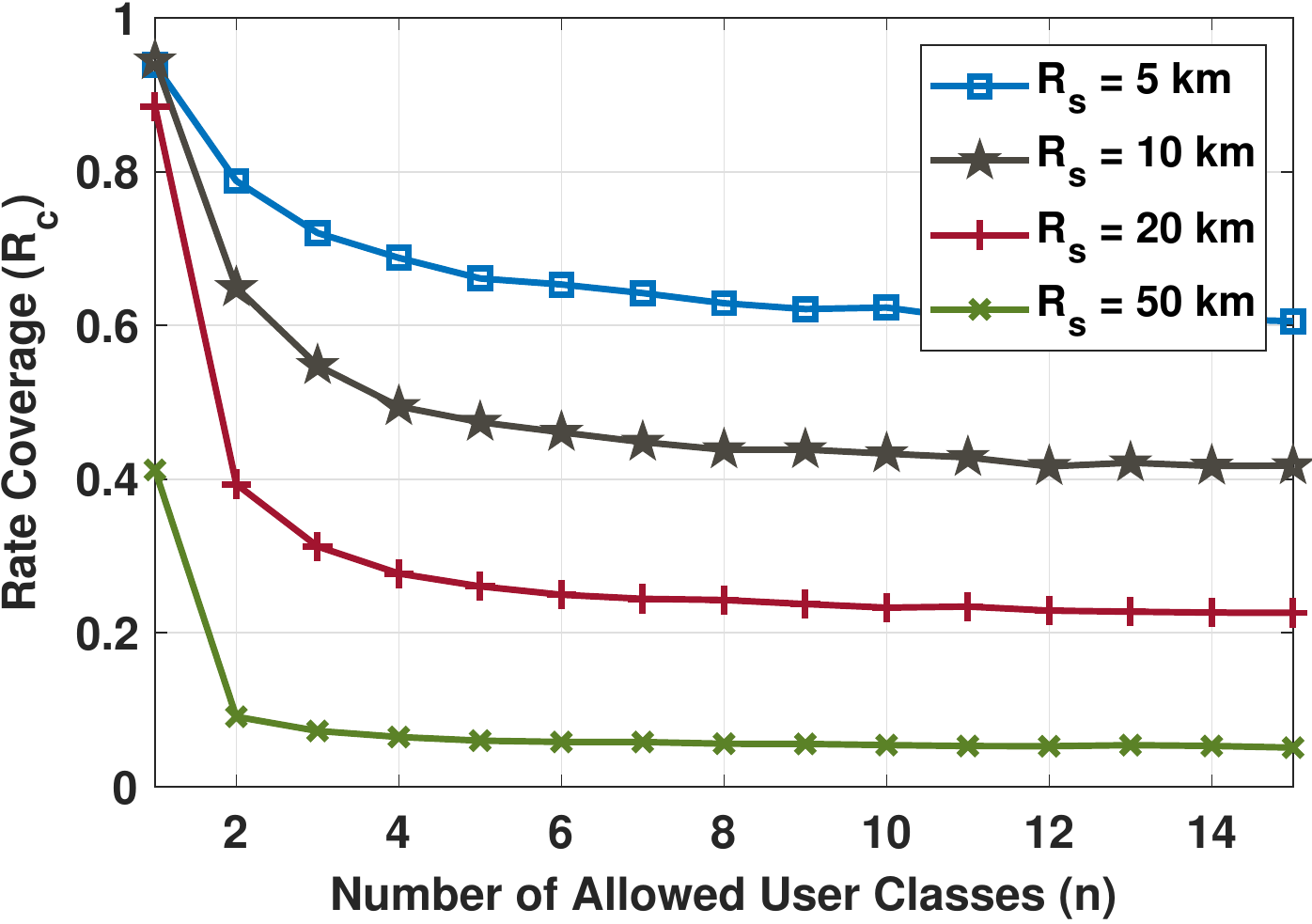}
   \caption{
   	Variation of SINR and rate coverage probability with respect to allowed number of user classes for different values of connectivity radius ($R_\sym$ in km) at a typical user in a broadcast system with geographically separated user classes. Content granularity is  $\Nc=15$, SINR threshold $\tau=10$ dB and rate threshold $\rho=15$ Mbps. Network density is 
$\lambda=.014$/km$^2$. The bandwidth varies with $R_\sym$ according to \eqref{eq:WTsymRelation1} with maximum value at 80 MHz. Other parameters are according to Table \ref{table:Tabel1}. The coverage decreases with $n$ due to increased interference at the typical user.  	
   	}
    \label{fig:5.2}
    \end{figure}
    
\noindent\textbf{Impact of $n$ on SINR and rate coverage probability:}
Fig. \ref{fig:5.2}(a) shows the variation of SINR coverage with $n$ for different values of connectivity  radius $R_\sym$. Here, SINR threshold  $\tau=10$ dB and there are $\Nc=15$ user classes. From Fig. \ref{fig:5.2}, we observe that for a fix value of $R_\sym$, SINR coverage decreases with increase in  $n$. This is due to the fact that more BBSs interfere as $n$ increases. However, after a certain $n$, the coverage doesn't changes much with $n$. This is because additional fraction of BBSs that interfere when $n$ increases by 1, is equal to $\frac{n}{n+1}-\frac{n-1}{n}=\frac1{n(n+1)}$ which decreases very fast with $n$. 
Therefore, after a certain $n$, there will not be a significant increase in the interference, which makes the  SINR constant with $n$.

It is also observed that the decrease in the SINR coverage probability is faster when $R_\sym$ is large. This can be justified as follows. First note that $n$ only affects the BBSs that can either be a interferer or a serving BBS, depending on the content they are showing. These BBSs lie in the ring of $R_\sym$ width denoting the region and their number approximately scales as  $\lambda\pi R_{\sym}^2$. Note that this number is large when $R_\sym$ is large. When we allow BBSs to show more advertisements/contents, a large number of these BBSs means that there is a larger number of potential interferes. 
When $R_\sym$ is small, there are less number of these potential  BBSs (or even 0), hence allowing more advertisements doesn't affect the coverage significantly.   
    Fig. \ref{fig:5.2}(b) shows the variation of rate coverage with respect to $n$ for different size of connectivity region. The rate coverage also follows similar behavior as SINR coverage as described in \eqref{eq:3.28_2}.

\begin{figure}[ht!]
    \centering

    \includegraphics[width=\plotsize]{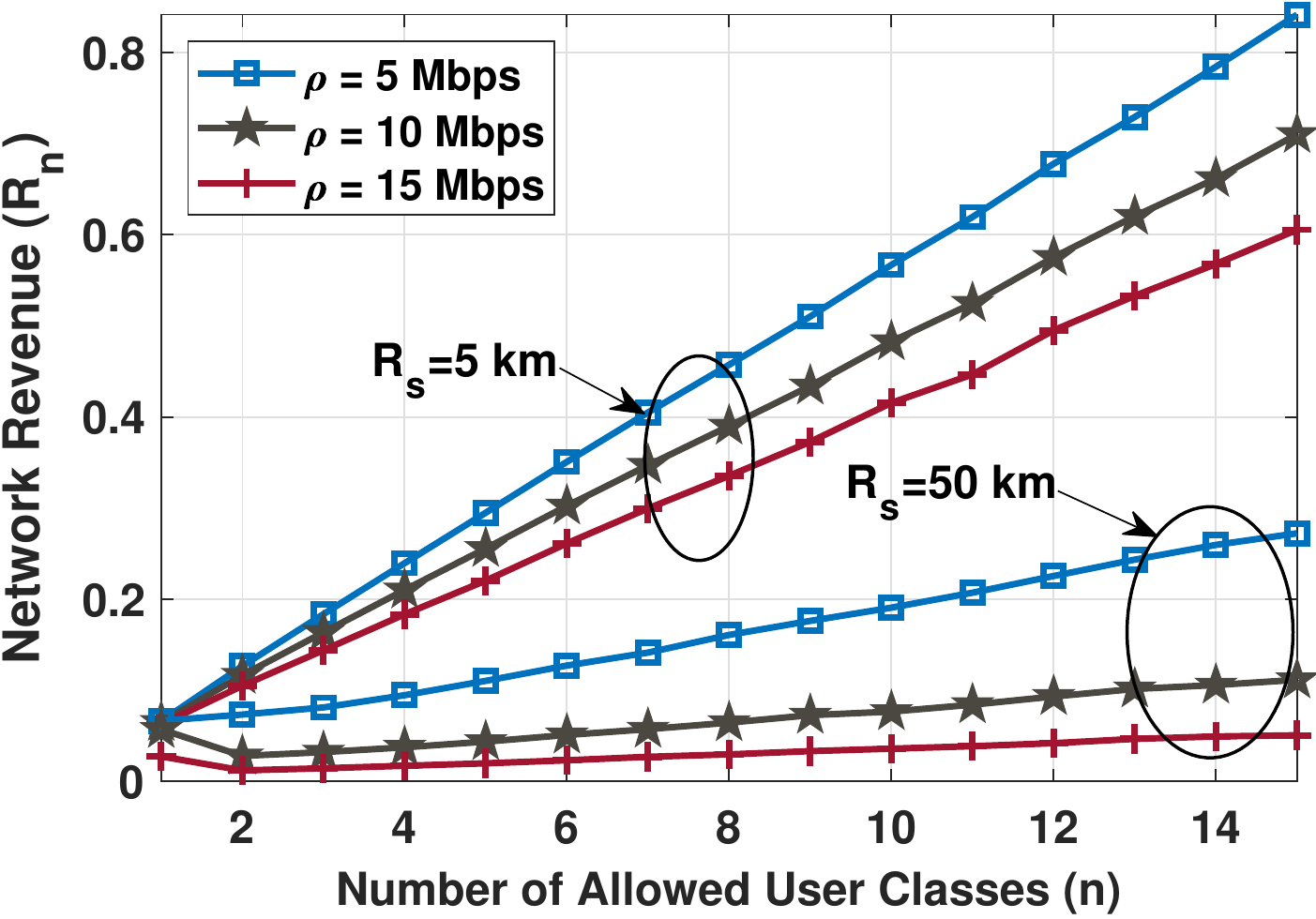}
   \caption{Variation of the total network revenue with respect to allowed number of user classes $n$ for different values of rate threshold $\rho$ (in Mbps) with $R_\sym=150\, \text{km}$  for a broadcast system with geographically separated user classes. Content granularity is  $\Nc=15$. Here, bandwidth varies with $R_\sym$ according to \eqref{eq:WTsymRelation1}. Other parameters are according to Table \ref{table:Tabel1}. 
   	}
    \label{fig:5.5}
    \end{figure}

\noindent\textbf{Impact of $n$ on the total network revenue:}    
Fig. \ref{fig:5.5} shows the impact of  $n$ on the total network revenue for different values of rate threshold.   Increase in $n$ means  catering to more number of user classes.  From Fig. \ref{fig:5.5}, we can observe that the revenue initially decreases  and then, increases with $n$.  The initial decrease in the revenue seen from $n=1$ to $n=2$ for some configurations is due to the decrease in rate coverage from $n=1$ to $n=2$, as observed in Fig. \ref{fig:5.2}(b) which dominated the increase in the revenue generated  due to catering to an additional user class.
 However, this behavior may depend on the target rate threshold and the value of $R_\sym$. 
   
\begin{figure}[ht!]
    \centering

    \includegraphics[width=\plotsize]{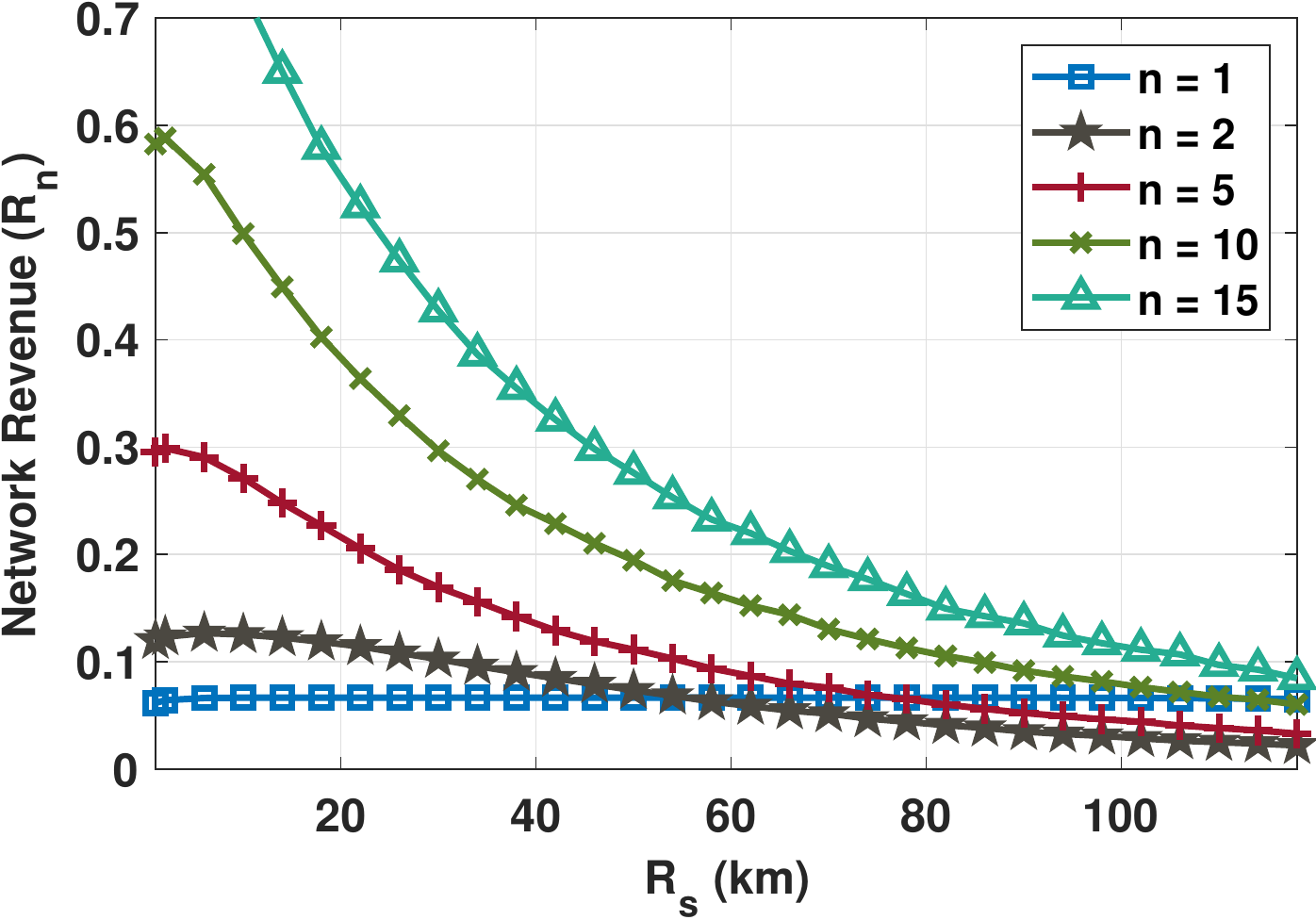}
   \caption{ Variation of the total network revenue with respect to the connectivity radius $R_\sym$ (in km) for different values of $n$ with $\rho=5\,\text{Mbps}$ for a broadcast system with geographically separated user classes. Content granularity is  $\Nc=15$. Here, bandwidth varies with $R_\sym$ according to \eqref{eq:WTsymRelation1} with maximum value at 50 MHz. Other parameters are according to Table \ref{table:Tabel1}. 
   	 }
    \label{fig:5.6}
    \end{figure}

\noindent    \textbf{Impact of $R_\sym$ on the total network revenue} 
  Fig. \ref{fig:5.6} shows the total revenue with respect to size of connectivity region $R_\sym$ for different $n$ with $\rho=10\,\text{Mbps}$ and $\Nc=15$.
We can observe that, the revenue decreases with increase in $R_\sym$. Also the revenue increases with $n$ for lower values of $R_\sym$. For a higher values of $R_\sym$, the behavior with $n$ is not monotonic. The revenue for $n\geq1$ may fall below the revenue for $n=1$. As discussed previously, the rate coverage decreases drastically from $n=1$ to $n=2$ for large $R_\sym$ which can dominate  the increase in the revenue generated  due to catering to an additional user class.
    
 \section{Conclusions}
 In this paper, we presented an analytical framework for the system performance of a broadcast network using stochastic geometry. Since all BBSs in the broadcast network are transmitting the same signal, signals from multiple BSs can be used to improve the coverage. We show that there exists a region such that all BBSs lying in this region may contribute to the desired signal power. We computed the SINR and rate coverage probability for a typical user located at the origin. We validated our results using numerical analysis. Using these results, we found that there exists an optimal region size which maximized the rate coverage. 
When users consists of many user classes having heterogenous content preference, network can schedule content to maximize its revenue. We  presented an analytical model of revenue thus obtained from users.  The results are validated through numerical analysis. We also present the variation of total revenue with respect to various parameters including number of user classes to be catered, size of connectivity region, and rate threshold. We show how content quality also affects the network decision on variety of content shown by the operator.

 \appendices
 \section{Proof for Lemma \ref{lemma:1}}\label{proof1}
 Using \eqref{eq:2}, the Laplace transform of the sum interference $I|\X_0$ is given as
\begin{align}
&\laplace{I\,|\,\X_0}(s)= \expects{}{e^{-sI}|\X_0}=\expects{}{\expS{-s\sum_{\X_{j}\in\Phi\cap \complementT{\Ball(0,{X_0}+R_{\sym})}}\beta_{j}{\dist{\X_{j}}}^{-\alpha}}}\nonumber\\
 &\overset{(a)}=\exp{\left(-\lambda\int_{\Phi\cap\complementT{\Ball(0,{X_0}+R_{\sym})}}\left(1-\expects{\beta}{e^{-s\beta \dist{\X}^{-\alpha}}}\right)\dd \x\right)}\nonumber\\
&\stackrel{(b)}=\exp{\left(-2\pi\lambda\int_{X_0+R_{\sym}}^{\RLimit}\left(1-\expects{\beta}{e^{-s\beta r^{-\alpha}}}\right)r\,\dd r \right)}
\end{align}
where (a) is due to the probability generating functional (PGFL) of homogeneous PPP \cite{Primer2016} and (b) is due to conversion to polar coordinates. Now, since $\beta$'s are exponetially distributed, using their MGF, we get
\begin{align}
&\laplace{I}(s)=\exp{\left(-2\pi\lambda\int_{\rcls+R_{\sym}}^{\RLimit}\frac{sr^{-\alpha}}{1+sr^{-\alpha}}r\,\dd r \right)}. \label{eq:3.10}
 \end{align}
 
 Similarly, the Laplace transform of desired signal power conditioned on the nearest is BS located at $\X_0$, is given as: 
 \begin{align}
    	&\mathcal{L}_{\Prx\,\vert\,\X_{0}}(s)=\mathbb{E}_{\Prx\given}\left[e^{-s\Prx}|\X_0\right]\nonumber\\
   	=&\expects{\beta_i,\X_i}{\expS{-s\left(\beta_0 \dist{ \X_0}^{-\alpha}+\sum_{ \X_{i}\in\Phi\cap
		\Ball {(\origin,X_0+R_{\sym})}\setminus\X_0}\beta_{i}\dist{\X_{i}}^{-\alpha}\right)}}\nonumber\\
	=&
	\expects{ \Phi\,\given}{\expects{\beta_0\given}{e^{-s\beta_0 \dist{\X_0}^{-\alpha}}}
		\right.\nonumber\\
	 &\times \left.\prod_{ \X_{i}\in\Phi\cap \Ball {(\origin,X_0+R_{\sym})}\setminus\X_0}
		\expects{\beta_i\vert X_i,\X_{0}}{e^{-s\beta_{i}\dist{ \X_{i}}^{-\alpha})}}}\label{eq:app1:imstep1}
\end{align}
Now, from Slivnyak theorem \cite{Primer2016}, we know that 	conditioned on $\X_0$, $\{\X_i:\X_i\in\Phi\}\setminus \{\X_0\}$ is a PPP with the same density. Therefore, using the PGFL of a PPP and noting that $\beta_i$'s are exponential RVs, we get 
\begin{align}
    	\mathcal{L}_{\Prx\,\vert\,\X_{0}}(s)=&\frac{1}{1+s\dist{\X_{0}} ^{-\alpha}}\exp\left(-2\pi\lambda \int_{X_0}^{X_0+R_{\sym}}\left(1-\frac{1}{1+sr^{-\alpha}}\right)r\dd r \right)\nonumber\\
    	=&\frac{1}{1+sX_0 ^{-\alpha}}\exp\left(-2\pi\lambda \int_{X_0}^{X_0+R_{\sym}}\frac{sr^{-\alpha}}{1+sr^{-\alpha}}r\dd r \right)\label{eq:3.9}.
\end{align}

\section{Proof of Theorem \ref{thm1}}\label{proof_thm1}

Using Lemma \ref{lemma:1}, we get 
  \begin{align}
  \laplace{\Prx\givenm=\rcls}(-jt)
 &\iftoggle{SingleColumn}{}{\nonumber\\&}
 =\frac{1+jt\rcls ^{-\alpha}}{1+t^{2}\rcls ^{-2\alpha}}\exp\left(-2\pi\lambda \int_{\rcls}^{\rcls+R_{\sym}}\frac{-jtr^{-\alpha}}{1-jtr^{-\alpha}}r\dd r \right)\label{app:eq:3_10}\\
\laplace{I\givenm=\rcls}(jt\tau)&=\exp{\left(-2\pi\lambda \int_{\rcls+R_{\sym}}^{\RLimit}\frac{jt\tau r^{-\alpha}}{1+jt\tau r^{-\alpha}}r\,\dd r\right)}
\label{app:eq:3.8}.
\end{align}
Substituting the above values in \eqref{eq:3.7}, we get
\begin{align}
&\mathbb{P}\left(\Prx>(I+\sigma^{2})\tau\givenm=\rcls\right)\nonumber\\
&=\frac{1}{2}+ \frac{1}{\pi}\int_{0}^{\infty}\frac{1}{t}\Im{\left(\frac{1+jt\rcls^{-\alpha}}{1+t^{2}\rcls^{-2\alpha}}\right)e^{-jt\tau\sigma^2}\times\right.\nonumber\\
&\hspace{5mm}\left.\exp\left(-2\pi\lambda\left(\int_{\rcls+R_{\sym}}^{\RLimit}\frac{jt\tau r^{-\alpha}}{1+jt\tau r^{-\alpha}}r\,\dd r
\iftoggle{SingleColumn}{}{\right.\right.\right.\nonumber\\&\hspace{10mm}\left.\left.\left.}
+\int_{\rcls}^{\rcls+R_{\sym}}\frac{-jtr^{-\alpha}}{1-jtr^{-\alpha}}r\,\dd r\right)
\right)}\dd t\nonumber
\end{align}
\begin{align}
&=\frac{1}{2}+ \frac{1}{\pi}\int_{0}^{\infty}\frac{1}{t}\Im{\left(\frac{1+jt\rcls^{-\alpha}}{1+t^{2}\rcls^{-2\alpha}}\right)e^{-jt\tau\sigma^2}\times\right.\nonumber\\
&\left.\exp\left(
	-2\pi\lambda\left(
		\int_{\rcls+R_{\sym}}^{\RLimit}\frac{t^2\tau^2 r^{-2\alpha+1}}{1+t^2\tau^2r^{-2\alpha}}\dd r
		+
		\int_{\rcls}^{\rcls+R_{\sym}}\frac{t^2r^{-2\alpha+1}}{1+t^2r^{-2\alpha}}\dd r
	\right.
\right.\right.\nonumber\\
&\left.\left.\left.+j\int_{\rcls+R_{\sym}}^{\RLimit}
\frac{t\tau r^{-\alpha+1}}{1+t^2\tau^2r^{-2\alpha}}\dd r\,-\,j \int_{\rcls}^{\rcls+R_{\sym}}\frac{tr^{-\alpha+1}}{1+t^2r^{-2\alpha}}\dd r\right)\right)}\dd t\label{eq:3.13}
\end{align}
where the last step is obtained using multiplication of conjugate terms. 
Now, if we define
\begin{align}
    &M(t,\rcls)
    \iftoggle{SingleColumn}{}{\nonumber\\&}
    =2\alpha{ t^{-2/\alpha}}\left[\int_{\rcls}^{\rcls+R_{\sym}}\frac{t^2r^{-2\alpha+1}}{1+t^2r^{-2\alpha}}\dd r+
\,\int_{\rcls+R_{\sym}}^{\RLimit}\frac{t^2\tau^2 r^{-2\alpha+1}}{1+t^2\tau^2r^{-2\alpha}}\dd r\right]\label{app1:mdef}\\
&N(t,\rcls)
\iftoggle{SingleColumn}{}{\nonumber\\&}
=2\alpha {t^{-2/\alpha}}\left[\int_{\rcls+R_{\sym}}^{\RLimit}\frac{t\tau r^{-\alpha+1}}{1+t^2\tau^2r^{-2\alpha}}\dd r\,-\,\int_{\rcls}^{\rcls+R_{\sym}}\frac{tr^{-\alpha+1}}{1+t^2r^{-2\alpha}}\dd r\right]\label{app1:ndef}
\end{align}
\eqref{eq:3.13} can be written as
\begin{align}
\iftoggle{SingleColumn}{}{&}\mathbb{P}\left(\Prx>(I+\sigma^{2})\tau\givenm=\rcls\right)
\iftoggle{SingleColumn}{&}{\nonumber\\&}
=\frac{1}{2}+ \frac{1}{\pi}\int_{0}^{\infty}\frac{1}{t}\Im{\left(\frac{1+jt\rcls^{-\alpha}}{1+t^{2}\rcls^{-2\alpha}}\right)\right.\nonumber\\
&\times\left.\exp\left(
	-2\pi\lambda 
	\addi{t^{2/\alpha}}M(t,\rcls)-j 2\pi\lambda\addi{ t^{2/\alpha}} N(t,\rcls)-jt\tau\sigma^2
	 \right)}\dd t\label{eq:3.13a}.
\end{align}
Now, with some trivial manipulations and substituting \eqref{eq:3.13a} in \eqref{eq:3.4}, we get 
\begin{align}
   \Pc(\tau, \lambda)
=&\frac{1}{2}+\frac{1}{\pi} \int_{0}^{\infty}\int_{0}^\infty2\pi\lambda \rcls e^{-\pi\lambda \rcls^2}\cdot\frac{1}{t}\cdot\left[\frac{1}{1+t^{2}\rcls ^{-2\alpha}}\right]
e^{-2\pi\lambda t^{2/\alpha}M(t,\,\rcls)/2\alpha}
\nonumber\\&
	\times\left[t\rcls^{-\alpha}\cos
		{\left(
			\frac{\pi}{\alpha}\lambda {t^{2/\alpha}}N(t,\,\rcls)+\,t\tau\sigma^2
		\right)}
\iftoggle{SingleColumn}{}{\right.\nonumber\\&\left.}
			-\sin{\left(\frac{\pi}{\alpha}\lambda{t^{2/\alpha}}N(t,\,\rcls)
			+\,t\tau\sigma^2\right)}
	\right]
\dd t\dd \rcls. \label{eq:3:27:old}
\end{align}
Further, the forms of $M$ and $N$ can be simplified by trivial manipulations and definition of incomplete Beta function in \eqref{app1:mdef} and \eqref{app1:ndef} to get
 \begin{align}
    &M(t,\rcls)=\bbfunc{
    						\frac1\alpha,
    						t^2(\rcls+R_{\sym})^{-2\alpha},
    						t^2\rcls^{-2\alpha}
    						} 
						\iftoggle{SingleColumn}{}{\nonumber
    						\\&
    						\qquad}
						+\tau ^{2/\alpha}
    						\bbfunc{
    						\frac1\alpha,
    						0,
    						t^2\tau^2(\rcls+R_{\sym})^{-2\alpha}
    						}, \text{ and }\label{thm1:mdef}\\
    &N(t,\,\rcls)=-\bbfunc{
    						\frac1\alpha+\frac12,
    						t^2(\rcls+R_{\sym})^{-2\alpha},
    						t^2\rcls^{-2\alpha}
    						} \iftoggle{SingleColumn}{}{\nonumber
    						\\&
    						\qquad}
					+\tau ^{2/\alpha}
    						\bbfunc{
    						\frac1\alpha+\frac12,
    						0,
    						t^2\tau^2(\rcls+R_{\sym})^{-2\alpha}
    						}
    						\label{thm1:ndef}
\end{align}
Now, we can substitute
\begin{align}
t&\rightarrow s /{(\lambda\pi)}^{\alpha/2}& 
u&\rightarrow v/\sqrt{\lambda\pi}\
\end{align}
in \eqref{eq:3:27:old}, \eqref{thm1:mdef} and \eqref{thm1:ndef} to get the desired result.

\section{Proof of Lemma \ref{lemma2}}\label{proof2}

From \eqref{eq:5.3}, the Laplace transform of desired signal power $\Prx$ conditioned that the nearest BS located at $\X_0$ is given as: 
 \begin{align}
    &\mathcal{L}_{\Prx\given}(s)=\mathbb{E}_{\Prx\,\given}\left[e^{-s\Prx}\right]\nonumber\\
&=\expects{\{\beta_i,\X_i\}\given}{\exp\left(-s\beta_0 { X_0}^{-\alpha}
\iftoggle{SingleColumn}{}{\right.\right.\nonumber\\
&
\left.\left.}
-\sum_{ \X_{i}\in\Phi\cap\Ball {(\origin,X_0+R_{\sym})}\setminus\X_0}
			s\beta_{i}{X_{i}}^{-\alpha}
			\indecator{M_i=M_0}\right)}\nonumber
\end{align}
which is similar to \eqref{eq:app1:imstep1} except the fact that the summation in the last term is over only those points that satisfy an additional condition $M_i=M_0$. From the independent thinning theorem, these points also form a PPP with density $\lambda \prob{M_i=M_0}=\lambda p$. Now using the PGFL of this PPP, we get		
\begin{align}
    \mathcal{L}_{\Prx\given}(s) 
    &=\frac{1}{1+sX_0^{-\alpha}}\exp\left(-2\pi\lambda p \int_{X_0}^{X_0+R_{\sym}}\frac{sr^{-\alpha}}{1+sr^{-\alpha}}r\dd r \right)\label{eq:5.20}. 
\end{align}
Now,  from \eqref{eq:5.2}, the Laplace transform of sum interference is 
\begin{align}
	&\laplace{I}(s)= \expects{I}{e^{-sI}}\nonumber\\
	&=\expects{I}{
		\exp\left(
		-\sum_{\X_{i}\in\Phi\cap \Ball(0,X_0+R_{\sym})\setminus \X_0}
						  s\beta_{i}{{X_{i}}}^{-\alpha}\indecator{M_i\ne M_0}
						  \iftoggle{SingleColumn}{}{\right.\right.\nonumber\\
				&\hspace{10mm}\left.\left.}
					-\sum_{\X_{j}\in\Phi\cap\complementT{\Ball(0,X_0+R_{\sym})}}
						s\beta_{j}{X_{j}}^{-\alpha}
					\right)}\nonumber\\
	&\stackrel{(a)}{=}
			\exp\left(-2\pi\lambda (1-p)\int_{X_0}^{X_0+R_{\sym}}
			\left(1-\expects{\beta}{e^{-s\beta r^{-\alpha}}}\right)r \dd r
			\iftoggle{SingleColumn}{}{\right.\nonumber\\
	&\left.\hspace{24mm}}
		-2\pi\lambda\int_{X_0+R_{\sym}}^{\RLimit}
			\left(1-\expects{\beta}{e^{-s\beta r^{-\alpha}}}\right)r\dd r
   	 \right)\nonumber\\
&\overset{(b)}=\exp\left(-2\pi\lambda(1-p)\int_{\rcls}^{X_0+R_{\sym}}\frac{sr^{-\alpha}}{1+sr^{-\alpha}}r\dd r
\iftoggle{SingleColumn}{}{\right.\nonumber\\&\hspace{25mm}\left.}
-2\pi\lambda\int_{X_0+R_{\sym}}^{\RLimit}\frac{sr^{-\alpha}}{1+sr^{-\alpha}}r\dd r \right),\label{eq:5.16}
\end{align}
 where (a) is due to the probability generating functional of homogeneous PPP and independent thinning theorem and (b) is due to MGF of exponentially distributed $\beta_i$'s.

\section{Proof of Theorem \ref{thm2}}\label{proof_thm2}
 Using Lemma \ref{lemma2}, we get
\begin{align}
    &\laplace{\Prx\givenm=\rcls}(-jt)
    \iftoggle{SingleColumn}{}{\nonumber\\
    &}
	=\frac{1+jt\rcls ^{-\alpha}}{1+t^{2}\rcls ^{-2\alpha}}\exp\left(-2\pi\lambda p \int_{\rcls}^{\rcls+R_{\sym}}\frac{-jtr^{-\alpha}}{1-jtr^{-\alpha}}rdr \right)\label{eq:5.22}\\
&\laplace{I\givenm=\rcls}(jt\tau)=\exp\left(-2\pi\lambda\left[(1-p)\int_{\rcls}^{\rcls+R_{\sym}}\frac{jt\tau r^{-\alpha}}{1+jt\tau r^{-\alpha}}r\,\dd r
\iftoggle{SingleColumn}{}{\right.\right.\nonumber\\&\hspace{20mm}
\left.\left.}+ \int_{\rcls+R_{\sym}}^{\RLimit}\frac{jt\tau r^{-\alpha}}{1+jt\tau r^{-\alpha}}r\,\dd r\right]\right)\label{eq:5.17}.
\end{align}
Substituting the above values in \eqref{eq:3.7}, we get
\begin{align}
&
\mathbb{P}\left(\Prx>(I+\sigma^{2})\tau\givenm=\rcls\right)
\nonumber\\
=&\frac{1}{2}+ \frac{1}{\pi} \int_{0}^{\infty}\frac{1}{t}\Im{\left(\frac{1+jt\rcls^{-\alpha}}{1+t^{2}\rcls ^{-2\alpha}}\right)e^{-jt\tau\sigma^2}\cdot\right.\nonumber\\
&\left.
\exp\left(-2\pi\lambda\left(\int_{\rcls+R_{\sym}}^{\RLimit}\frac{jt\tau r^{-\alpha}}{1+jt\tau r^{-\alpha}}r\,\dd r+
\iftoggle{SingleColumn}{}{\right.\right.\right.\nonumber\\&\indent\left.\left.\left.}
(1-p)\int_{\rcls}^{\rcls+R_{\sym}}\frac{jt\tau r^{-\alpha}}{1+jt\tau r^{-\alpha}}r\,\dd r
\right.\right.\right.\nonumber\\&\indent\left.\left.\left.
+p\int_{\rcls}^{\rcls+R_{\sym}}\frac{-jtr^{-\alpha}}{1-jtr^{-\alpha}}r\,\dd r\right)\right)
}\dd t \nonumber
\\=&\frac{1}{2}+ \frac{1}{\pi} 
\int_{0}^{\infty}\frac{1}{t}\Im{\left(\frac{1+jt\rcls^{-\alpha}}{1+t^{2}\rcls ^{-2\alpha}}\right)e^{-jt\tau\sigma^2}\cdot\right.\nonumber\\&
\left.
\exp\left(-2\pi\lambda\left(\int_{\rcls}^{\rcls+R_{\sym}}\frac{p\,t^2r^{-2\alpha+1}}{1+t^2r^{-2\alpha}}\dd r\,
\iftoggle{SingleColumn}{}{\right.\right.\right.\nonumber\\&\left.\left.\left.+\,}
\int_{\rcls+R_{\sym}}^{\RLimit}\frac{t^2\tau^2 r^{-2\alpha+1}}{1+t^2\tau^2r^{-2\alpha}}\dd r\,+\,
\int_{\rcls}^{\rcls+R_{\sym}}\frac{(1-p)t^2\tau^2 r^{-2\alpha+1}}{1+t^2\tau^2r^{-2\alpha}}\dd r\right.\right.\right.\nonumber\\&\left.\left.\left.+\,j\left[\int_{\rcls+R_{\sym}}^{\RLimit}\frac{t\tau r^{-\alpha+1}}{1+t^2\tau^2r^{-2\alpha}}\dd r\,+\,
\int_{\rcls}^{\rcls+R_{\sym}}\frac{(1-p)\,t\tau r^{-\alpha+1}}{1+t^2\tau^2r^{-2\alpha}}\dd r\,\iftoggle{SingleColumn}{}{\right.\right.\right.\right.\nonumber\\&\left.\left.\left.\left.}
-\,\int_{\rcls}^{\rcls+R_{\sym}}\frac{p\,tr^{-\alpha+1}}{1+t^2r^{-2\alpha}}\dd r\right]\right)\right)}\dd t\label{eq:5.23}
\end{align}
where the last step is obtained using multiplication of conjugate terms and rearranging into the real and imaginary parts. 
Now, if we define 
\begin{align}
  M'(t,\rcls)
  \iftoggle{SingleColumn}{}{\nonumber\\
}=&
2{\alpha t^{-\alpha/2}}\left[ \int_{\rcls}^{\rcls+R_{\sym}}
\frac{p\,t^2r^{-2\alpha+1}}{1+t^2r^{-2\alpha}}\dd r \right.\nonumber\\&
\left.+\int_{\rcls+R_{\sym}}^{R}
  \frac{t^2\tau^2 r^{-2\alpha+1}}{1+t^2\tau^2r^{-2\alpha}}\dd r+\int_{\rcls}^{\rcls+R_{\sym}}\frac{(1-p)t^2\tau^2 r^{-2\alpha+1}}{1+t^2\tau^2r^{-2\alpha}}\dd r\right],\label{app4:mdef}\\
    N'(t,\rcls)
    \iftoggle{SingleColumn}{}{\\
}=&
2\alpha {t^{-\alpha/2}}\left[\int_{\rcls+R_{\sym}}^{R}\frac{t\tau r^{-\alpha+1}}{1+t^2\tau^2r^{-2\alpha}}\dd r\right.\nonumber\\&+\left.\int_{\rcls}^{\rcls+R_{\sym}}\frac{(1-p)\,t\tau r^{-\alpha+1}}{1+t^2\tau^2r^{-2\alpha}}\dd r-\int_{\rcls}^{\rcls+R_{\sym}}\frac{p\,tr^{-\alpha+1}}{1+t^2r^{-2\alpha}}\dd r\right],\label{app4:ndef}
\end{align}
\eqref{eq:5.23} can be written as
\begin{align}
\mathbb{P}\left(\Prx>(I+\sigma^{2})\tau\givenm=\rcls\right)&\iftoggle{SingleColumn}{}{\nonumber\\
&}
=\frac{1}{2}+ \frac{1}{\pi}\int_{0}^{\infty}\frac{1}{t}\Im{\left(\frac{1+jt\rcls^{-\alpha}}{1+t^{2}\rcls^{-2\alpha}}\right)\right.\nonumber\\
&\times\left.\exp\left(
	-2\pi\lambda 
	\addi{t^{2/\alpha}}M'(t,\rcls)-j 2\pi\lambda\addi{ t^{2/\alpha}} N'(t,\rcls)-jt\tau\sigma^2
	 \right)}\dd t\label{eq:5.23a}.
\end{align}
Now, with some trivial manipulations and substituting \eqref{eq:5.23a} in \eqref{eq:3.4}, we get 
\begin{align}
 \Pc(\tau, \lambda)=&\frac{1}{2}+\frac{1}{\pi} \int_{0}^{\infty}\int_{0}^\infty2\pi\lambda \rcls e^{-\pi\lambda \rcls^2}\cdot\frac{1}{t}
\cdot\left[\frac{1}{1+t^{2}\rcls ^{-2\alpha}}\right]\cdot
e^{-
\frac\pi\alpha
\lambda t^{2/\alpha}M'(t,\,\rcls)}
\nonumber\\&
\times \left[t\rcls^{-\alpha}\cos{\left(\frac\pi\alpha\lambda {t^{2/\alpha}} N'(t,\rcls)+\,t\tau\sigma^2\right)}
 \iftoggle{SingleColumn}{}{\right.\nonumber\\&\hspace{-12mm}\left.}
-\sin{\left(\frac\pi\alpha\lambda {t^{2/\alpha}} N'(t,\rcls)+\,t\tau\sigma^2\right)}\right]
\dd t\dd \rcls.\label{eq:5.33_old}
\end{align}
Further, the forms of $M$ and $N''$ can be simplified 
by trivial manipulations and definition of incomplete Beta function in \eqref{app4:mdef} and \eqref{app4:ndef} to get
%
%
%
 \begin{align}
&M'(t,\rcls)=p\bbfunc{\frac1{\alpha},t^2(\rcls+R_{\sym})^{-2\alpha},t^2\rcls^{-2\alpha}}\nonumber
\\
&
+\tau^{2/\alpha}(1-p)
\bbfunc{\frac1{\alpha},
(t\tau)^2(\rcls+R_{\sym})^{-2\alpha},(t\tau)^2\rcls^{-2\alpha}}
+\tau ^{2/\alpha}
\bbfunc{\frac1{\alpha},
0,
t^2\tau^2(\rcls+R_{\sym})^{-2\alpha}}
.\label{thm2:mdef}
\end{align}
    \text{and,}
     \begin{align}
&N'(t,\rcls)=-p\bbfunc{\frac1{\alpha}+\frac12,t^2(\rcls+R_{\sym})^{-2\alpha},t^2\rcls^{-2\alpha}}\nonumber
\\
&
+\tau^{2/\alpha}(1-p)
\bbfunc{\frac1{\alpha}+\frac12,
(t\tau)^2(\rcls+R_{\sym})^{-2\alpha},(t\tau)^2\rcls^{-2\alpha}}
+\tau ^{2/\alpha}
\bbfunc{\frac1{\alpha}+\frac12,
0,
t^2\tau^2(\rcls+R_{\sym})^{-2\alpha}}
.\label{thm2:ndef}
\end{align}
Now, we can substitute
\begin{align}
t&\rightarrow s /{(\lambda\pi)}^{\alpha/2}& 
u&\rightarrow v/\sqrt{\lambda\pi}\
\end{align}
in \eqref{eq:5.33_old}, \eqref{thm2:mdef} and \eqref{thm2:ndef} to get the desired result.

\bibliographystyle{IEEEtran}

\end{document}

%% file: defis1.tex
\usepackage{bbm}
\usepackage{amssymb}
\usepackage{amsmath,amsthm}
\usepackage{color}
\usepackage{amsfonts}
\usepackage{graphicx}
\usepackage{verbatim}
\usepackage{epstopdf}
\usepackage{cite}

\newcommand{\bbfunc}[1]{Q\left(#1\right)}

\newcommand{\R}{\mathbb{R}}

\newcommand{\indecator}[1]{\mathbbm{1}\left(#1\right)}

\newcommand{\rcv}{\mathrm{r_c}}

\newcommand{\dist}[1]{\|#1\|}

\newcommand{\origin}{\mathbf{o}}

\renewcommand{\Im}[1]{\mathsf{Im}\left[#1\right]}

\newcommand{\dd}{\mathrm{d}}

\newcommand\prob[1]{\mathbb{P}\left[#1\right]}

\newcommand\indside[1]{\mathbbm{1}\left({#1}\right)}

\newcommand{\SINR}{\text{SINR}}

\newcommand{\Rate}{\text{Rate}}

\newcommand{\Ball}{\mathcal{B}}
\newcommand{\expects}[2]{\mathbb{E}_{#1}\left[#2\right] }

\newcommand{\laplace}[1]{\mathcal{L}_{#1} }
\newcommand{\ie}{{\em i.e.}~}

\newcommand{\Pc}{\mathrm{p_c}}
\newcommand{\Rc}{\mathrm{r_c}}
\newcommand{\rcls}{\mathrm{r_0}}
\newcommand{\given}{|\,\X_0}
\newcommand{\givenm}{|\,X_0}

\newcommand{\tx}{\mathsf{t~\!}}

\newcommand{\x}{{\mathbf{x}}}

\newcommand{\X}{\mathbf{X}}

\newcommand{\expS}[1]{\exp{\left(#1\right)}}

\newcommand{\betafunc}[3]{\mathsf{B}\left(#2,#3;#1\right)}
\def\home{\hbox{\kern3pt \vbox to13pt{}%
   \pdfliteral{q 0 0 m 0 5 l 5 10 l 10 5 l 10 0 l 7 0 l 7 5 l 3 5 l 3 0 l f
               1 j 1 J -2 5 m 5 12 l 12 5 l S Q }%
   \kern 13pt}}

\newcommand{\nf}{a}

\newcommand{\msqrt}{\overline{m}}

\newcommand{\DeltaT}{\Delta}
\renewcommand{\rcls}{u}
\newcommand{\RLimit}{\infty}
\newcommand{\Prx}{P_\rx}
\renewcommand{\Prx}{S}
\renewcommand{\SINR}{\mathtt{SINR}}

\renewcommand{\Rate}{\mathtt{Rate}}
\newcommand{\Nc}{N_\mathrm{c}}

\newcommand{\Rn}{\mathrm{R}_\mathrm{n}}
\newcommand{\selset}{\mathcal{M}}

\newcommand{\complementT}[1]{#1^{\complement}}
\newcommand{\set}[1]{\mathsf{#1}}
\newtheorem{theorem}{Theorem}
\newtheorem{lemma}{Lemma}

\newtheorem{remark}{Remark}

%% file: PBroadcastCovv2.bbl
\begin{thebibliography}{10}
\providecommand{\url}[1]{#1}
\csname url@samestyle\endcsname
\providecommand{\newblock}{\relax}
\providecommand{\bibinfo}[2]{#2}
\providecommand{\BIBentrySTDinterwordspacing}{\spaceskip=0pt\relax}
\providecommand{\BIBentryALTinterwordstretchfactor}{4}
\providecommand{\BIBentryALTinterwordspacing}{\spaceskip=\fontdimen2\font plus
\BIBentryALTinterwordstretchfactor\fontdimen3\font minus
  \fontdimen4\font\relax}
\providecommand{\BIBforeignlanguage}[2]{{%
\expandafter\ifx\csname l@#1\endcsname\relax
\typeout{** WARNING: IEEEtran.bst: No hyphenation pattern has been}%
\typeout{** loaded for the language `#1'. Using the pattern for}%
\typeout{** the default language instead.}%
\else
\language=\csname l@#1\endcsname
\fi
#2}}
\providecommand{\BIBdecl}{\relax}
\BIBdecl

\bibitem{ahu2}
M.~El-Hajjar and L.~Hanzo., ``A survey of digital television broadcast
  transmission techniques,'' \emph{{IEEE Communication Surveys and Tutorials}},
  vol.~15, no.~4, p. 1924–1949., 2013.

\bibitem{ahu3}
{L. Fay and L. Michael and D. Gmez-Barquero and N. Ammar and M.W. Caldwell},
  ``An overview of the {ATSC} 3.0 physical layer specification,'' \emph{IEEE
  Transactions on Broadcasting}, vol.~62, pp. 159--171, 2016.

\bibitem{ahu4}
D.~Gómez-Barquero, C.~Douillard, P.~Moss, and V.~Mignone, ``{DVB-NGH}: The
  next generation of digital broadcast services to handheld devices,''
  \emph{IEEE Transactions on Broadcasting}, vol.~60, no.~2, pp. 246--257, 2014.

\bibitem{ahu19}
{D. {Catrein}, J. {Huschke}, U. {Horn}}, ``Analytic evaluation of a hybrid
  broadcast-unicast {TV} offering,'' in \emph{Proc. IEEE Vehicular Technology
  Conference-Spring (VTC-spring)}, 2008, pp. 2864--2868.

\bibitem{ahu14}
L.~Rong, S.~E. Elayoubi, and O.~B. Haddada, ``Performance evaluation of
  cellular networks offering {TV} services,'' \emph{{IEEE Transactions on
  Vehicular Technology}}, vol.~60, no.~2, 2011.

\bibitem{JGA2011}
J.~G. Andrews, F.~Baccelli, and R.~K. Ganti, ``A tractable approach to coverage
  and rate in cellular networks,'' \emph{IEEE Transactions on communications},
  vol.~59, no.~11, pp. 3122--3134, 2011.

\bibitem{Primer2016}
J.~G. Andrews, A.~K. Gupta, and H.~S. Dhillon, ``A primer on cellular network
  analysis using stochastic geometry,'' \emph{{arXiv preprint
  arXiv:1604.03183}}, 2016.

\bibitem{ahu7}
M.~Haenggi, \emph{Stochastic Geometry for Wireless Networks}.\hskip 1em plus
  0.5em minus 0.4em\relax Cambridge: Cambridge University Press, 2012.

\bibitem{Guo2015}
A.~Guo and M.~Haenggi, ``Asymptotic deployment gain: A simple approach to
  characterize the {SINR} distribution in general cellular networks,''
  \emph{IEEE Transactions on Communications}, vol.~63, pp. 962--976, Mar. 2015.

\bibitem{AndrewsmmWaveTut2016}
J.~G. Andrews, T.~Bai, M.~N. Kulkarni, A.~Alkhateeb, A.~K. Gupta, and R.~W.
  Heath, ``Modeling and analyzing millimeter wave cellular systems,''
  \emph{IEEE Trans. Commun.}, vol.~65, no.~1, pp. 403--430, Jan 2017.

\bibitem{ahu8}
{W. Lu and M. Di Renzo}, ``Stochastic geometry modeling of cellular networks:
  Analysis, simulation and experimental validation,'' in \emph{{Proc. ACM
  International Conference on Modeling, Analysis and Simulation of Wireless and
  Mobile Systems (MSWiM)}}, 2015, p. 179–188.

\bibitem{ahu9}
{A. {Shokair} and Y. {Nasser} and M. {Crussière} and J. {Hélard} and O.
  {Bazzi}}, ``Analytical study of the probability of coverage in hybrid
  broadcast-unicast networks,'' in \emph{Proc. Annual International Symposium
  on Personal, Indoor and Mobile Radio Communications (PIMRC)}, 2018, pp. 1--6.

\bibitem{SaChGup2020}
R.~{Sahu}, K.~K. {Chaurasia}, and A.~K. {Gupta}, ``{SINR} and rate coverage of
  broadcast networks using stochastic geometry,'' in \emph{Proc. International
  Conference on Signal Processing and Communications (SPCOM)}, July 2020, pp.
  1--5.

\bibitem{ahu23}
L.~Peterson and R.~Groot, ``Location-based advertising: The key to unlocking
  the most value in the mobile advertising and location-based services
  markets,'' \emph{Peterson Mobility Solutions}, 2009.

\bibitem{ahu21}
L.~Shi, E.~Obregon, K.~W. Sung, J.~Zander, and J.~Bostrom, ``{CellTV}—on the
  benefit of {TV} distribution over cellular networks: A case study,''
  \emph{IEEE Transactions on Broadcasting}, vol.~60, no.~1, pp. 73--84, 2014.

\bibitem{ahu11}
J.~Gil-Pelaez, ``{Note on the inversion theorem},'' \emph{Biometrika}, vol.
  38,no. 3-4, p. 481–482, 1951.

\bibitem{ITUP1546}
{International Telecommunication Union}, ``Recommendation{ ITU-R P.1546-6},
  method for point-to-area predictions for terrestrial services in the
  frequency range 30 mhz to 4 000 {MHz},'' \emph{P Series Radiowave
  Propagation}, August 2019.

\bibitem{DBTdotPDF}
------, \emph{Handbook on Digital Terrestrial Television Broadcasting Networks
  and Systems Implementation}, 2016.

\bibitem{GupZhaAnd2015}
A.~K. {Gupta}, X.~{Zhang}, and J.~G. {Andrews}, ``{SINR} and throughput scaling
  in ultradense urban cellular networks,'' \emph{IEEE Wireless Communications
  Letters}, vol.~4, no.~6, pp. 605--608, 2015.

\end{thebibliography}
